\documentclass[acmsmall]{acmart}
\usepackage[utf8]{inputenc}
\usepackage[T1]{fontenc}
\usepackage{microtype}
\usepackage{graphicx} %
\usepackage[colorinlistoftodos]{todonotes}
\usepackage{bbold}
\usepackage{algorithm}
\usepackage{algpseudocode}
\usepackage{mathtools}
\usepackage{caption}
\usepackage{subcaption}
\usepackage{siunitx}

\newcommand{\den}[1]{\llbracket #1 \rrbracket}

\newcommand{\bfalse}{\mathbf{f}}
\newcommand{\btrue}{\mathbf{t}}
\newcommand{\indic}[1]{\mathbb{1}\pbrack{#1}}
\newcommand{\seq}{\,;\,}

\newcommand{\kw}[1]{\ensuremath{\operatorname{#1}}}

\DeclarePairedDelimiter{\abs}{\vert}{\vert}
\DeclarePairedDelimiter{\pparen}{\lparen}{\rparen}
\DeclarePairedDelimiter{\pbrack}{\lbrack}{\rbrack}
\DeclarePairedDelimiter{\pbrace}{\lbrace}{\rbrace}

\setcopyright{rightsretained}
\acmDOI{10.1145/3632858}
\acmYear{2024}
\copyrightyear{2024}
\acmSubmissionID{popl24main-p80-p}
\acmJournal{PACMPL}
\acmVolume{8}
\acmNumber{POPL}
\acmArticle{16}
\acmMonth{1}
\received{2023-07-11}
\received[accepted]{2023-11-07}

\begin{document}
\title{Optimal Program Synthesis via Abstract Interpretation}

\author{Stephen Mell}
\orcid{0009-0003-7469-8974}
\affiliation{%
  \institution{University of Pennsylvania}
  \country{USA}
}
\email{sm1@cis.upenn.edu}

\author{Steve Zdancewic}
\orcid{0000-0002-3516-1512}
\affiliation{%
  \institution{University of Pennsylvania}
  \country{USA}
}
\email{stevez@cis.upenn.edu}

\author{Osbert Bastani}
\orcid{0000-0001-9990-7566}
\affiliation{%
  \institution{University of Pennsylvania}
  \country{USA}
}
\email{obastani@seas.upenn.edu}
\begin{abstract}
We consider the problem of synthesizing 
programs with numerical constants that optimize a quantitative objective, such as accuracy, over a set of input-output examples.
We propose a general framework for optimal synthesis of such programs in a given domain specific language (DSL), with provable optimality guarantees. Our framework enumerates programs in a general search graph, where nodes represent subsets of concrete programs. To improve scalability, it uses $A^*$ search in conjunction with a search heuristic based on abstract interpretation; intuitively, this heuristic establishes upper bounds on the value of subtrees in the search graph, enabling the synthesizer to identify and prune subtrees that are provably suboptimal. In addition, we propose a natural strategy for constructing abstract transformers for monotonic semantics, which is a common property for components in
DSLs for data classification.
Finally, we implement our approach in the context of two such existing DSLs, demonstrating that our algorithm is more scalable than existing optimal synthesizers.
\end{abstract}



\maketitle

\section{Introduction}
Due to their interpretability, robustness, and data-efficiency, there is recent interest in synthesizing programs to solve data processing and querying tasks, including handling semi-structured and unstructured data such as images and natural language text. Examples include \emph{neurosymbolic programs} that incorporate deep neural network (DNN) components to extract semantic information from raw data~\citep{shah2020learning,chen2021web}, as well as \emph{fuzzy matching programs} that use predicates with quantitative semantics to approximately match real-valued data~\cite{mell2023synthesizing}.
For instance, \citet{shah2020learning} synthesizes programs that label sequence data, \citet{mell2023synthesizing} synthesizes 
queries over trajectories output by an object tracker, and \citet{chen2021web} synthesizes web question answering programs.
Most work focuses on programming by example (PBE), where the user provides a set of input-output (IO) examples, and the goal is to synthesize a program that generates the correct output for each input. 
There are two key properties that distinguish synthesis of such
programs from traditional PBE:
\begin{itemize}
\item \textbf{Quantitative objectives:} The goal in neurosymbolic synthesis is typically to optimize a quantitative objective such as accuracy or $F_1$ score rather than to identify a program that is correct on all examples (which may be impossible).
\item \textbf{Numerical constants:} Programs operating on fuzzy real-world data or the outputs of DNN components typically include real-valued constants that serve as thresholds;
for example,
when querying video trajectories, one constant may be a threshold on the maximum velocity of the object.
\end{itemize}
While these properties occasionally arise in traditional PBE settings (e.g., minimizing resource consumption), they are fundamental issues in neurosymbolic synthesis.
Furthermore, in the neurosymbolic setting, there is often additional structure
that can be exploited to improve synthesis performance---for instance, some of the numerical components
might be
monotone
in their inputs.

Most existing systems focus on synthesizing examples in a particular domain-specific language (DSL).
In these settings, prior work has leveraged monotonicity of the semantics to prune the search space~\citep{chen2021web,mell2023synthesizing}. One general framework is \citet{shah2020learning}, which uses \emph{neural relaxations} to guide search over a general DSL.
At a high level, they use $A^*$ search to enumerate partial programs in the DSL, which are represented as a directed acyclic graph (DAG). In general, $A^*$ search prioritizes the order in which to enumerate partial programs based on a score function (called a \emph{heuristic}) that maps each partial program to a real-valued score. When the heuristic is \emph{admissible}---i.e., its output is an upper bound on the objective value for \emph{any} completion of that partial program (assuming the goal is to maximize the objective)---then $A^*$ search is guaranteed to find the optimal program
if it terminates.\footnote{We mean optimal on the given IO examples. This notion ignores suboptimality due to generalization error, which can be handled using standard techniques from learning theory~\citep{kearns1994introduction}.}

\citet{shah2020learning} proposes the following heuristic: fill each hole in the partial program with an untrained DNN, and then maximize the quantitative objective as a function of the DNN parameters using resulting objective value as the heuristic. However, this score function is only guaranteed to be admissible under assumptions that typically do not hold in practice: (i) the neural relaxations are sufficiently expressive to represent any program in the DSL, which requires very large DNNs, and (ii) maximization of the DNN parameters converges to the global optimum, which does not hold for typical strategies such as stochastic gradient descent (SGD). Furthermore, SGD cannot handle non-differentiable objectives, which include common objectives such as accuracy and $F_1$ score.

Thus, a natural question is whether we can construct practical heuristics that are guaranteed to be admissible. In this work, we take inspiration from \emph{deduction-guided synthesis}, which uses automated reasoning techniques such as SMT solvers~\citep{sketch2006,gulwani2011synthesis,bornholt2016optimizing} or abstract interpretation~\citep{cousot1977abstract} to prune partial programs from the search space---i.e., prove that no completion of a given partial program can satisfy the given IO examples. In particular, we propose using abstract interpretation to construct heuristics for
synthesis for quantitative objectives. Traditionally, abstract interpretation can be used to prune partial programs by replacing each hole with an abstract value overapproximating all possible concrete values that can be taken by that hole in the context of a given input. Then, if the abstract output does not include the corresponding concrete output, that partial program cannot possibly be completed into a program that satisfies that IO example, so it can be pruned.

Our key insight is that
abstract interpretation can similarly be used to construct
an admissible heuristic for a quantitative objective. Essentially, we can use abstract interpretation to overapproximate the possible objective values obtained by any completion of a given partial program; then, the supremum of concrete values represented by the abstract output serves as an upper bound of the objective, so it can be used as an admissible heuristic. Thus, given abstract transformers for the DSL components and for the quantitative objective, our framework can synthesize optimal
programs.

In addition, we propose general strategies for constructing abstract domains and transformers for common DSLs and objectives.
As discussed above, many DSLs have monotone components.
In these settings, a natural choice of abstract domain is to use intervals for the real-valued constants; then, a natural abstract transformer is to evaluate the concrete semantics on the upper and lower bounds of the intervals. This strategy can straightforwardly be shown to correctly overapproximate the concrete semantics.

We
implement our approach in the context of two
DSLs---namely, the NEAR~\cite{shah2020learning} DSL for the CRIM13~\cite{crim13} dataset, and the Quivr~\cite{mell2023synthesizing} DSL and benchmark. In our experiments, we demonstrate that our approach significantly outperforms an adaptation of Metasketches~\cite{bornholt2016optimizing}---an existing optimal synthesis framework based on SMT solving---to our setting, as well as an ablation that uses breadth first search instead of $A^*$ search. Our approach significantly outperforms both of these baselines in terms of running time. In summary, our contributions are:
\begin{itemize}
\item We propose a novel algorithm for optimal synthesis
which performs enumerative search over a space of \emph{generalized partial programs}. To prioritize search, it uses the $A^*$ algorithm with a search heuristic based on abstract interpretation. If it returns a program, then that program is guaranteed to be optimal (Section~\ref{sec:framework}).
\item In practice, many
DSLs have types that are equipped with a partial order, such as the real numbers or Booleans. For these types,
we propose to use intervals as the abstract domains.
For monotone DSL components---i.e., the concrete semantics respects the partial orders---a natural choice of abstract transformer is to simply apply the concrete semantics to the lower and upper bounds of the interval (Section~\ref{sec:monotone}).
\item We implement our framework in the context of two
existing DSLs (Section~\ref{sec:impl}) and show in our experiments that it outperforms Metasketches~\cite{bornholt2016optimizing}, a state-of-the-art optimal synthesis technique based on SMT solvers, and a baseline that uses breadth-first search instead of our search heuristic (Section~\ref{sec:exp}).
\end{itemize}

\section{Motivating Example}
\label{sec:motivating}

We consider a task where the goal is to synthesize a
program for predicting the behavior of mice based on a video of them interacting~\cite{shah2020learning},
motivated by a data analysis problem in biology. In particular, biologists use mice as model animals to investigate both basic biological processes and to develop new therapeutic interventions, which sometimes requires determining the effect of an intervention on mouse behavioral patterns, such as
the nature and duration of interactions with other mice. For example, Figure~\ref{fig:ill:videoframe} depicts two mice in an enclosure, engaging in the ``sniff'' behavior.

\begin{figure}[t]
\centering
\includegraphics[width=0.5\textwidth]{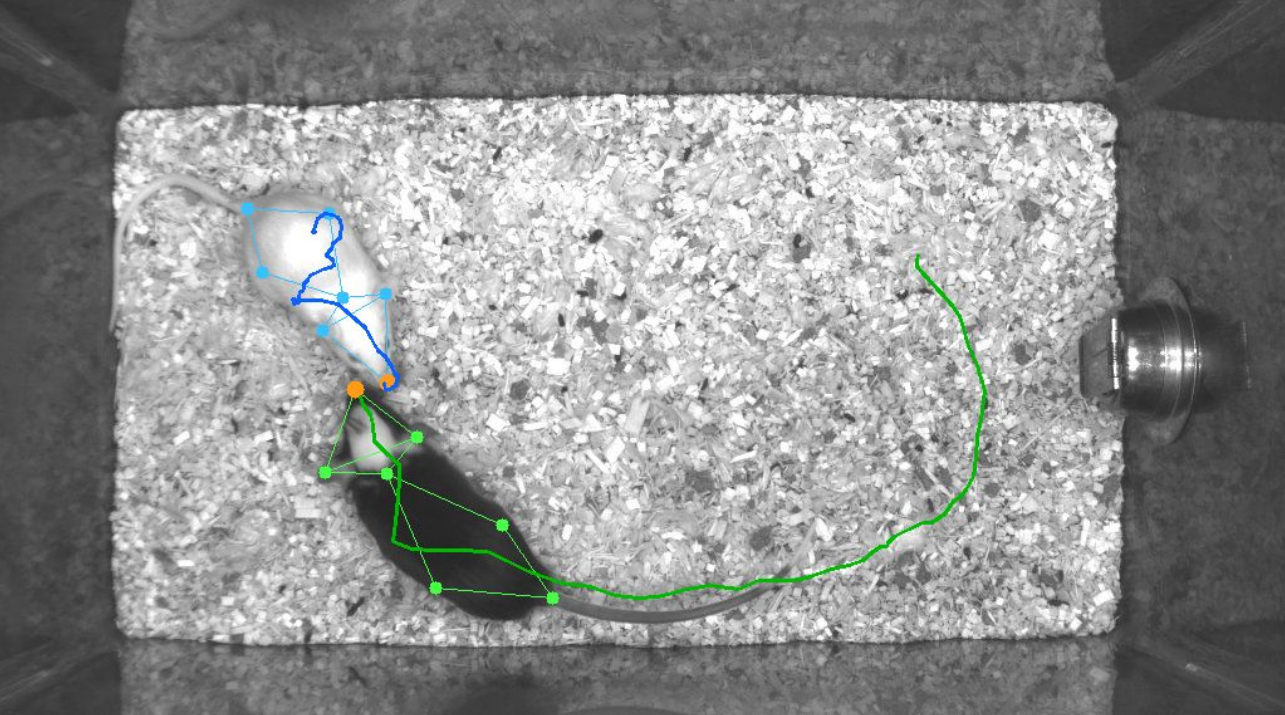}
\caption{A frame from a video of two mice interacting~\citep{calms21}; the mice are very close together, and are exhibiting the ``sniff'' behavior. The video has been processed using deep neural networks to produce certain keypoints, which are shown.}
\label{fig:ill:videoframe}
\end{figure}

Doing this behavior analysis typically involves researchers viewing and manually annotating these behaviors in hours of video, which is very labor intensive. As a result, program
synthesis has been applied to automating this task~\cite{shah2020learning}. Their approach first uses an object tracker to track each mouse across frames in the video, producing trajectories represented as a sequence of 2D positions for each mouse. Then, they featurize each step in the trajectory---for instance, if there are two mice in the video, then one feature might be the distance between them in each frame. Based on this sequence of features, the goal is to predict a label for the behavior (if any) that the mice are engaged in during each frame (producing a sequence of labels, one for each frame). \citet{shah2020learning} solves this problem by synthesizing a
program in
a functional DSL for processing trajectories (the NEAR DSL).
In summary, the goal is to synthesize a program that takes as input a sequence of feature vectors, and outputs a sequence of labels. We consider the programming by example (PBE) setting, where we are given a number of human annotated examples, and the goal is to synthesize a program that maximizes some objective, such as accuracy or $F_1$ score.

For example, consider synthesizing a program that, given a featurized trajectory representing two mice in a video, outputs the behavior of the mice at each time step. In particular, the input is a trajectory $x\in\mathbb{R}^*$ (where $T^*$ denotes lists of $T$s), where the feature $x[t]\in\mathbb{R}$ on time step $t$ encodes the distance between the two mice at that time step, and the output is $y\in\{\bfalse,\btrue\}^*$, where $y[t]$ encodes whether the mice are engaging in the ``sniff'' behavior ($y=\btrue$ if so, and $y=\bfalse$ if not) at time step $t$.

We consider synthesizing such a program based on a single training example $x_1=(101,65)$ and $y_1=(\bfalse,\btrue)$ (i.e., the first frame has mouse distance 101 and is labeled ``not sniff'' and the second frame has mouse distance 65 and is labeled ``sniff''). Our goal is to find some program that classifies a dataset of videos well---i.e., if we evaluate the program on the videos to get predicted labels, and then compute a classification metric (e.g. accuracy) between the predicted and true labels, the discrepancy should be small.
 Given a program $p$, its accuracy is $1$ if $\den{p}(x_1)=y_1$ and $1/2$ if $\den{p}(x_1)[0]=y_1[0]$ but $\den{p}(x_1)[1]\neq y_1[1]$ (or vice versa), and $0$ otherwise (where $x[i]$ is the $i$-th item in a sequence $x$).

Consider the candidate program ``$\operatorname{map}(d \leq 50)$''. In traditional syntax, this program would be ``$\operatorname{map}\ (\lambda d\,.\,d\leq50)\ x$'', but here the input is specified separately and the $\lambda$ omitted (in combinator-style, similar to the NEAR DSL).
This program performs a map over the sequence of frames, and for each one, it would output whether the mouse distance in that frame is less than or equal to $50$. For example, when evaluated on $x_1$, it outputs
\begin{align*}
\den{\operatorname{map}(d\leq50)}(x_1)=(\bfalse, \bfalse).
\end{align*}
Thus, its accuracy is $1/2$, since it correctly labels the first frame but not the second.

One strategy for computing the optimal program (i.e. the one that maximizes the objective) is to enumerate partial programs (i.e. programs with \emph{holes} representing pieces that need to filled to obtain a concrete program) in the DSL, evaluate the objective on every concrete program, and then choose the best one. There are several challenges to this approach:
\begin{itemize}
\item Unlike traditional synthesis, where we can stop enumerating when we reach a concrete program that satisfies the given specification, for optimal synthesis, we need to enumerate all programs or risk returning a suboptimal program.
\item Traditional synthesizers use a variety of techniques to prune the search space to improve scalability. For instance, they might use deduction to prove that no completion of a partial program can satisfy the specification---i.e., no matter how the holes in the partial program are filled, the specification will not hold. However, these techniques are not directly applicable to optimal synthesis.
\item Synthesizing real-valued constants poses a problem: one approach is to discretize the constants, enumerate all of them, and choose the best, but this approach can be prohibitively slow. For example, suppose we are enumerating completions of the partial program
\begin{align*}
\operatorname{map}(d \leq \;??),
\end{align*}
where $??$ is a \emph{hole} that needs to be filled with a real value $\theta \in [0, 100]$ to obtain a concrete program. If we discretize $\theta\in\{0,1,...,100\}$, then we would enumerate $\operatorname{map}(d \leq 0),...,\operatorname{map}(d \leq 100)$, evaluate each of these on $(x_1,y_1)$ to measure its accuracy, and choose the program with the highest accuracy.
\end{itemize}

Our framework uses two key innovations to address these challenges:
\begin{itemize}
\item \textbf{Generalized partial programs:} Our framework takes the traditional notion of partial programs, representing sets of concrete programs as completions of syntax with holes, and extends it to more general sets of programs, equipped with a directed acyclic graph (DAG) structure. This allows us to avoid discretization of real-valued constants and reduces the branching-factor of the search space.
\item $\mathbf{A^*}$ \textbf{search:} Rather than enumerate programs in an arbitrary order (e.g., breadth first search), our framework uses $A^*$ search to enumerate programs. This allows us to avoid considering all programs in the search space, similar to deductive pruning, while still returning the optimal program.
\end{itemize}
We describe each of these techniques in more detail below.

\subsubsection*{Generalized Partial Programs.}

Traditionally, the search space over partial programs is a DAG, where the nodes are partial programs, and there is an edge $\hat{p}\to\hat{p}'$ if $\hat{p}'$ can be obtained by filling a hole in $\hat{p}$ using some production in the DSL. For instance, there is an edge
\begin{align*}
\operatorname{map}(d\le\;??)\to\operatorname{map}(d\le50)
\end{align*}
in this DAG, since we have filled the hole with the value $50$. However, even if we discretize the search space, there are 101 ways to fill this hole. As a consequence, if even a single completion of $\operatorname{map}(d\le\;??)$ is valid, then we cannot prune it from the search space (even ignoring that we want the optimal program rather than just any valid program).

Instead, in our framework, we
allow search DAGs beyond just programs with holes, so long as
each node represents a set of concrete programs
and the node's children collectively represent
that set. As a practical instantiation of the general framework, we consider partial programs where holes for real-valued constants may be annotated with constraints on the value that can be used to fill them. For instance, the generalized partial program
\begin{align*}
\operatorname{map}(d \leq {??}_{[50,100]})
\end{align*}
represents the set of concrete programs
\begin{align*}
\pbrace*{\operatorname{map}(d \leq \theta) \mid \theta \in [50,100]}
\end{align*}
Then, the children of this generalized partial program in the search DAG should split the constraint in a way that covers the search space---e.g.,
\begin{align*}
\textsf{children}(\operatorname{map}(d \leq {??_{[50,100]}}))
=\{\operatorname{map}(d \leq {??_{[50,75]}}),\operatorname{map}(d \leq {??_{[75,100]}})\}.
\end{align*}
This strategy presents more opportunities for pruning the search space. For instance, even if we cannot prune the program $\operatorname{map}(d \leq {??_{[50,100]}})$, we may be able to prune the program $\operatorname{map}(d \leq {??_{[50,75]}})$. Then, rather than needing to enumerate 51 programs (i.e., one for each $\theta\in\{50,...,100\}$), we would only need to prune $\operatorname{map}(d \leq {??_{[50,75]}})$ and evaluate 26 programs (i.e., one for each $\theta\in\{75,...,100\}$). Of course, we can further subdivide the search space to further reduce enumeration.

\subsubsection*{$A^*$ Search.}

Next, we describe how we achieve the optimal synthesis analogue of pruning by using $A*$ search.
At a high level, $A^*$ search enumerates nodes in a search graph according to a \emph{heuristic}; for our purposes, a heuristic is function that maps a partial program $\hat{p}$ to a real value $\mu$, and a heuristic is said to be \emph{admissible} if it is an upper
bound on the best possible objective value for any completion of $\hat{p}$---i.e.,
\begin{align*}
p\in\textsf{completions}(\hat{p})\Rightarrow\mu\ge\textsf{objective}(p).
\end{align*}
The heuristic adapts deductive reasoning to optimal synthesis: whereas deductive reasoning guarantees that no completion of $\hat{p}$ can satisfy the given specification, the heuristic guarantees that no completion of $\hat{p}$ can achieve objective value greater than $\mu$---e.g., if we find a concrete program with objective value $\ge\mu$, we can safely prune completions of $\hat{p}$ from the search DAG.

While $A^*$ search has previously been used in
synthesis for quantitative objectives~\cite{shah2020learning}, the heuristics used are not admissible and so do not provide theoretical guarantees. Our key contribution is showing that abstract interpretation can naturally be adapted to design admissible heuristics. Abstract interpretation can be used for traditional deductive reasoning as follows: fill each hole in the current partial program with an abstract value $\top$, evaluate the partial program using abstract semantics, and check if the abstract output is consistent with the specification.

In our example, a natural choice of abstract domain for constants is the interval domain. In addition, rather than fill each hole with $\top$, if a hole has a constraint $??_{[a,b]}$, we can instead fill it with the interval $[a,b]$. For instance, for our example program $\operatorname{map}(d \leq \;??_{[50, 75]})$, we can fill the hole with the interval $[50,75]$ to obtain $\operatorname{map}(d\leq [50,75])$. Then, the abstract semantics $\den{\cdot}^\#$ evaluate as follows:
\begin{align}
\label{eqn:exabssem}
\den{\operatorname{map}(d \leq {??_{[50, 75]}})}^\#(x_1)
=(\den{101 \leq {??_{[50, 75]}}}^\#,\den{65 \leq {??_{[50, 75]}}}^\#)
=(\bfalse, \top).
\end{align}
In other words, the first element is $\bfalse$ since $x_1=(101,65)$, and we know $101\not\le\theta$ for any $\theta\in[50,75]$, and the second element is $\top$ since the relationship between $65$ and $\theta\in[50,75]$ can be either $\bfalse$ or $\btrue$. Importantly, here the abstract values are over holes in the program rather than over the input to it.

Traditionally, we would then check whether this abstract output is consistent with the specification. For optimal synthesis, we observe that we can define an abstract semantics for the objective function. In our example, we can compute an ``abstract accuracy'' as follows (where $\mathbb{1}$ is a indicator function, and $\leq$, $\mathbb{1}$, $+$, $\cdot$, and $=$ are abstracted in the obvious way):
\begin{align}
\label{eqn:exabsacc}
\begin{split}
\text{accuracy}^\#(\hat{p})
&=\frac{1}{2}\sum_{t=1}^2\mathbb{1}\left[\den{\operatorname*{map}(d\le\;??_{[50,75]})}^\#(x_1)[t]=y_1[t]\right] \\
&=\frac{1}{2}\left(\indic{\bfalse=\bfalse}+\indic{\top=\btrue}\right) \\
&=\frac{1}{2}([1,1]+[0,1]) \\
&=[1/2,1].
\end{split}
\end{align}
In other words, for the first frame, the concrete programs represented by $\operatorname{map}(d\le\;??_{[50,75]})$ all predict $\bfalse$, which equals $y_1[1]$, so this frame is always correctly classified. In contrast, for the second frame, concrete programs programs may output either $\bfalse$ or $\btrue$, so we are uncertain whether this frame is correctly classified. Thus, the true accuracy is in the interval $[1/2,1]$.  Since abstract interpretation is guaranteed to overapproximate the semantics, we can use the upper bound of the abstract objective value as our heuristic---e.g., for $\operatorname{map}(d\leq [50,75])$, this heuristic computes $\mu=1$.

Finally, we summarize the full search procedure starting from the node $\operatorname{map}(d \leq {??}_{[0,100]})$ in our search DAG, which has an abstract objective value of $[1/2, 1]$. First, we split it into
\begin{align*}
\hat{p}_1=\operatorname{map}(d \leq {??_{[0, 50]}})
\qquad\text{and}\qquad
\hat{p}_2=\operatorname{map}(d \leq {??_{[50, 100]}}).
\end{align*}
The abstract accuracies of $\hat{p}_1$ and $\hat{p}_2$ are $[1/2, 1/2]$ and $[1/2,1]$, so their heuristic values are $\mu_1=1/2$ and $\mu_2=1$, respectively. Since $\mu_2>\mu_1$, our algorithm explores $\hat{p}_2$ next, splitting it into
\begin{align*}
\hat{p}_3=\operatorname{map}(d \leq {??_{[50, 75]}})
\qquad\text{and}\qquad
\hat{p}_4=\operatorname{map}(d \leq {??_{[75, 100]}}),
\end{align*}
which have abstract accuracies of $[1/2, 1]$ and $[1,1]$, respectively, and heuristic values of $\mu_3=1$ and $\mu_4=1$, respectively. In this example, the lower bound on the accuracy of $\hat{p}_4$ is also $1$, so we know that any choice $\theta\in[75,100]$ for filling this hole is guaranteed to achieve an accuracy of $1$; thus, any concrete program $\operatorname{map}(d \leq \theta)$ such that $\theta\in[75,100]$ is optimal, and our algorithm can terminate without ever considering $\hat{p}_1$ or $\hat{p}_3$.

In general, our algorithm terminates once the range of possible optimal values is sufficiently small. For each node on the search frontier, we have an upper and lower bound on the objective value of all the programs it represents. The greatest of these lower bounds provides a lower bound on the best possible objective value, and the greatest of these upper bounds provides an upper bound on the best possible objective value. As a result, once the difference between the bounds is $\leq \epsilon$, we know that we have a program within $\epsilon$ of being optimal.

\section{Optimal Synthesis via Abstract Interpretation}
\label{sec:framework}

In this section,
we consider the program synthesis problem where: (i) programs in the domain-specific language may have real-valued constants, and (ii) the synthesis objective is real-valued, where the goal is to return the optimal program (Section~\ref{sec:problem}). Then, we describe our synthesis algorithm for solving this problem, which uses $A^*$ search in conjunction with a search heuristic based on abstract interpretation (Section~\ref{sec:astarsynth}).

\subsection{Problem Formulation}
\label{sec:problem}

\subsubsection*{Domain-Specific Language.}

For concreteness, consider a DSL whose syntax is given by a context-free grammar $\mathcal{G}=(V,\Sigma,R,P)$, where $V$ is the set of nonterminals, $\Sigma$ is the set of terminals, $P\in V$ is the start symbol, and $R$ is the set of productions
\begin{align*}
P \Coloneqq X\mid c \mid f(P,...,P) \qquad (c\in\mathcal{C}, f\in\mathcal{F})
\end{align*}
where $X$ is a symbol representing the input, $c\in\mathcal{C}$ is a constant (including real-valued constants, i.e., $\mathbb{R}\subseteq\mathcal{C}$),
and $f\in\mathcal{F}$ is a DSL component (i.e., function), but our framework extends straightforwardly to more general grammars. We let $p\in\mathcal{P}=\mathcal{L}(\mathcal{G})\subseteq\Sigma^*$ denote the concrete programs in our DSL. Furthermore, we assume the DSL has denotational semantics $\den{\cdot}$, where $\den{p}:\mathcal{X}\to\mathcal{Y}$ maps inputs $x\in\mathcal{X}$ to outputs $y\in\mathcal{Y}$ according to the following rules:
\begin{align*}
\den{X}(x)=x,\quad
\den{c}(x)=c,\quad
\den{f(p_1,...,p_k)}(x)=f(\den{p_1}(x),...,\den{p_k}(x)),
\end{align*}
where we assume the functions $f:\mathcal{X}_1\times...\times\mathcal{X}_k\to\mathcal{Y}$ are given.

In Section~\ref{sec:motivating}, we considered programs like $\operatorname{map}(d\le50)$, which we will use as a running example in this section. They are simplified versions of programs from the NEAR DSL, and are generated by the grammar (where $d$ represents the input)
\begin{align*}
    E \Coloneqq d \mid c \mid \operatorname{map}(E) \mid E \leq E \qquad (c\in\mathbb{R})
\end{align*}

\subsubsection*{Task Specification.}

We consider programming by example (PBE), where each task is specified by a sequence
$Z \in \mathcal{Z} \coloneqq (\mathcal{X} \times \mathcal{Y})^*$ of input-output (IO) examples, and the goal is to compute a program $p^*$ that maximizes a given quantitative objective %
$\phi:\mathcal{P}\times\mathcal{Z}\to\mathbb{R}$
\begin{align*}
p^*\in\operatorname*{\arg\max}_{p\in\mathcal{P}}\phi(p,Z).
\end{align*}
 $\phi$ is a function of the semantics applied to the examples $(x,y)\in Z$---i.e., there is a function $\phi_0:(\mathcal{Y}\times\mathcal{Y})^* \to \mathbb{R}$ such that
\begin{align}
\label{eqn:objectivedecomp}
\phi(p,Z)=\phi_0\left(\{(\den{p}(x),y)\}_{(x,y)\in Z}\right),
\end{align}
In our running example, we choose $\phi_0$ as follows, so $\phi(p,Z)$ is the accuracy of $P$ on $Z$:
\begin{align*}
\phi_0(W)=\frac{1}{|W|}\sum_{(y',y)\in W}\indic{y'=y},
\end{align*}
where $\indic{\cdot}$ is the indicator function.

\subsubsection*{Partial Programs.}

A common strategy for PBE is to enumerate \emph{partial programs}---programs
in the DSL that have holes---to try and find
a concrete program that satisfies the given IO examples. Intuitively, partial programs are partial derivations in the grammar $\mathcal{G}$. To formalize this notion, given two sequences
${\hat{p},\hat{p}'}\in(\Sigma\cup V)^*$, we write $\hat{p}\to\hat{p}'$ if $\hat{p}'$ can be obtained by replacing a nonterminal symbol $N_i\in V$ in $\hat{p}$ with the right-hand side of a production $r=N_i\to M_1...M_k\in R$ for that nonterminal---i.e. $\hat{p}=N_1...N_i...N_h$ and $\hat{p}'=N_1...M_1...M_k...N_h$. We denote this relationship by $\hat{p}'=\text{fill}(\hat{p},i,r)$---i.e. we obtain $\hat{p}'$ by filling the $i$th hole $N_i$ in $\hat{p}$ using production $r$. Next, we write $\hat{p}\xrightarrow{*}\hat{p}'$ if there exists a sequence $\hat{p}=\hat{p}_1\to...\to\hat{p}_n=\hat{p}'$, and say $\hat{p}'$ can be \emph{derived} from $\hat{p}$.

Note that concrete programs $p\in\mathcal{P}$ are sequences $p\in\Sigma^*$ that can be derived from the start symbol $P$ (i.e. $P\xrightarrow{*}p$). Similarly, a partial program is a sequence
$\hat{p}\in\hat{\mathcal{P}} \coloneqq (\Sigma\cup V)^*$ that can be derived from $P$---i.e. $P\xrightarrow{*}\hat{p}$. The only difference is that $\hat{p}$ may contain nonterminals, which are called \emph{holes}. The space of partial programs naturally forms a directed acyclic graph (DAG) via the relation $\hat{p}\to\hat{p}'$; note that concrete programs are leaf nodes in this DAG. Thus, we can perform synthesis by enumerating partial programs according to the structure of this DAG. Furthermore, given a partial program $\hat{p}\in\hat{\mathcal{P}}$ and a concrete program $p\in\mathcal{P}$, we say $p$ is a \emph{completion} of $\hat{p}$ if $\hat{p}\xrightarrow{*}p$ (i.e., $p$ can be obtained from $\hat{p}$ by iteratively filling the holes of $\hat{p}$).

In our running example, traditional partial programs correspond to the grammar
\begin{align*}
    \bar{E} \Coloneqq {??} \mid d \mid c \mid \operatorname{map}(\bar{E}) \mid \bar{E} \leq \bar{E} \qquad (c\in\mathbb{R}).
\end{align*}

\subsubsection*{Generalized Partial Programs.}

A key challenge is searching over real-valued constants $c\in\mathbb{R}$.
Our grammar
contains infinitely many productions of the form $P\to\theta$ for $\theta\in\mathbb{R}$, and even
if we discretize this search space, the number of productions is still large in practice.

To address this challenge, we propose a strategy where we enumerate \emph{generalized partial programs} $\hat{p}\in\hat{\mathcal{P}}$, which generalize (i) the fact that partial programs correspond to sets of concrete programs (i.e. the set of their completions), and (ii) the DAG structure of partial programs. 
\begin{definition}
\label{def:searchdag}
\rm
A space of \emph{generalized partial programs} is a set $\hat{\mathcal{P}}$ together with a \emph{concretization function} $\gamma:\hat{\mathcal{P}}\to2^{\mathcal{P}}$ and a \emph{DAG structure} $\text{children}:\hat{\mathcal{P}}\to2^{\hat{\mathcal{P}}}$, such that
\begin{align}
\label{eqn:dagcond}
\gamma(\hat{p})=\bigcup_{\hat{p}'\in\text{children}(\hat{p})}\gamma(\hat{p}').
\end{align}
\end{definition}
Intuitively, $\gamma(\hat{p})\subseteq\mathcal{P}$ is the set of concrete programs represented by the abstract program $\hat{p}$. In addition, $\text{children}$ encodes a DAG structure on $\hat{\mathcal{P}}$ that is compatible with $\gamma$---i.e., the children $\hat{p}'$ of $\hat{p}$ must collectively contain all the concrete programs in $\hat{p}$.

For example, to capture traditional partial programs, we let
\begin{align*}
\gamma(\hat{p})=\{p\in\mathcal{P}\mid\hat{p}\xrightarrow{*}p\},
\end{align*}
i.e. $p$ is a completion of $\hat{p}$, and
\begin{align*}
\text{children}(\hat{p})=\{\hat{p}'\in\hat{\mathcal{P}}\mid\hat{p}\to\hat{p}'\}.
\end{align*}
In Section~\ref{sec:monotone}, we will propose generalized partial programs that can include constraints on real-valued holes---e.g., $??_{[0,1]}$ is a partial program that can only be filled by a real value $\theta\in[0,1]$.

Finally, a simple way to satisfy Definition~\ref{def:searchdag} is to define $\gamma$ based on the $\text{children}$ function---i.e., we can define $p\in\gamma(\hat{p})$ if there exists a sequence $\hat{p}=\hat{p}_1,...,\hat{p}_n=p$ 
of generalized partial programs such that $\hat{p}_{j+1}\in\text{children}(\hat{p}_j)$ for all $1<j<n$. In other words, we can reach the concrete program $p$ from the generalized partial program $\hat{p}$ in the search DAG. This strategy straightforwardly guarantees (\ref{eqn:dagcond}), since by definition, every $p\in\gamma(\hat{p})$ must be the descendant of some child of $\hat{p}$.

In our running example, the generalized partial programs correspond to the grammar
\begin{align*}
    \hat{E} \Coloneqq {??}_{[a,b]} \mid {??} \mid d \mid c \mid \operatorname{map}(\hat{E}) \mid \hat{E} \leq \hat{E} \qquad (a,b\in\mathbb{R}),
\end{align*}
while the concretization function satisfies
\begin{align*}
    \gamma(\operatorname{map}({??})) &= \{\operatorname{map}(p) \mid p \in \mathcal{P}\} \\
    \gamma(\operatorname{map}(d \leq {??}_{[50,100]})) &= \{\operatorname{map}(d \leq c) \mid a \leq c \leq b\}
\end{align*}
and the children function satisfies
\begin{align*}
    \operatorname{children}(\operatorname{map}({??})) &= \{\operatorname{map}({??}_{[0,100]}), \operatorname{map}(d), \operatorname{map}(c), \operatorname{map}(\operatorname{map}({??})), \operatorname{map}({??} \leq {??})\} \\
    \operatorname{children}(\operatorname{map}(d \leq {??}_{[50,100]})) &= \{\operatorname{map}(d \leq {??}_{[50,75]}), \operatorname{map}(d \leq {??}_{[75,100]})\}.
\end{align*}

\subsubsection*{Abstract Objective.}

For now, we consider an \emph{abstract objective} $\den{\cdot}_{\phi}^\#$ that directly maps generalized partial programs to abstract real values; it is typically constructed compositionally by providing abstract transformers for each component $f\in\mathcal{F}$ as well as for the objective function $\phi_0$, and then composing them together.
In particular, the abstract objective has type $\den{\hat{p}}_{\phi}^\#:\mathcal{Z}\to\hat{\mathbb{R}}$, where $\hat{\mathbb{R}}$ is an abstract domain for the reals representing the potential objective values $\phi(p,Z)$ (e.g. the interval domain). Rather than require a concretization function for $\hat{\mathbb{R}}$, we only need an upper bound for this abstract domain---i.e., a function $\mu:\hat{\mathbb{R}}\to\mathbb{R}$, which encodes the intuition that ``$\mu(\hat r)$ is larger than any real number $r$ contained in $\hat{r}$''.
\begin{definition}
\rm
Given objective $\phi$ and generalized partial programs $(\hat{\mathcal{P}},\gamma,\text{children})$, an abstract objective $(\den{\cdot}_{\phi}^\#,\mu)$ is \emph{valid} if
\begin{align}
\label{eqn:abscond1}
\big(p\in\gamma(\hat{p})\big)\Rightarrow\big(\mu(\den{\hat{p}}_{\phi}^\#(Z))\ge\phi(p,Z)\big)
\qquad(\forall p\in\mathcal{P}),
\end{align}
and
\begin{align}
\label{eqn:abscond3}
\mu(\den{p}_{\phi}^\#(Z))=\phi(p,Z)
\qquad(\forall p\in\mathcal{P}).
\end{align}
\end{definition}
Intuitively, (\ref{eqn:abscond1}) says that $\mu(\den{\hat{p}}_{\phi}^\#(Z))$ is an upper bound on the objective value $\phi(p,Z)$ for concrete programs $p$ is contained in the abstract program $\hat{p}$. In addition, (\ref{eqn:abscond3}) says that for concrete programs, the abstract objective and concrete objective coincide.

Finally, a typical choice of $\hat{\mathbb{R}}$ is the space of intervals $\hat{\mathbb{R}}$,
where $(r,r')\in\hat{\mathbb{R}}$ represents the set of real numbers $\{r''\in\mathbb{R}\mid r\le r''\le r'\}$ (see Definition~\ref{defn:intervaldomain} for details). Then, the upper bound is given by $\mu((r,r'))=r'$. The abstract objective $\den{\cdot}_{\phi}^\#$ depends on the DSL; we describe a general construction in Section~\ref{sec:monotone}.

In our running example, the objective $\phi$ decomposes according to Equation~\ref{eqn:objectivedecomp} into a semantics $\den{\cdot}$ and accuracy $\phi_0$. We thus define the abstract objective $\den{\cdot}_{\phi}^\#$ in terms of an abstract semantics $\den{\cdot}^\# : \hat{\mathcal{P}} \to \hat{\mathbb{R}}$ (Equation~\ref{eqn:exabssem}), and abstract accuracy (Equation~\ref{eqn:exabsacc}).

\subsection{$A^*$ Synthesis via Abstract Interpretation}
\label{sec:astarsynth}
\begin{algorithm}[t]
\begin{algorithmic}
\Procedure{$A^*$-Synthesis}{$Z,\mathcal{G},\hat{\mathcal{P}},\gamma,\den{\cdot}_{\phi}^\#,\mu$}
\State $h\gets\text{heap}(\operatorname{sort\_by}=\lambda \hat{p}. \mu(\den{\hat{p}}_{\phi}^\#(Z)))$
\State $h.\text{push}(\hat{p}_0)$
\While{\textbf{true}}
\State $\hat{p}\gets h.\text{pop\_greatest}()$
\If{$\max\{\mu(\den{\hat{p}'}_{\phi}^\#(Z)) \mid \hat{p}' \in h\} - \max\{\nu(\den{\hat{p}'}_{\phi}^\#(Z)) \mid \hat{p}' \in h\} \leq \varepsilon$}
\State \Return $\arg\max\{\nu(\den{\hat{p}'}_{\phi}^\#(Z)) \mid \hat{p}' \in h\}$
\EndIf
\For{$\hat{p}'\in\text{children}(\hat{p})$}
\State $h.\text{push}(\hat{p}')$
\EndFor
\EndWhile
\EndProcedure
\end{algorithmic}
\caption{Our algorithm takes as input a task specification $Z$, along with a DSL $\mathcal{G}$, a space of generalized programs $(\hat{\mathcal{P}},\gamma)$, abstract objective $(\den{\cdot}_{\phi}^\#,\mu)$, objective lower bound $\nu$, objective error tolerance $\varepsilon$, and returns the optimal program $p^*$ for task $Z$. To do so, it uses the abstract objective
as a heuristic in $A^*$ search, starting from the
initial generalized partial program $\hat{p}_0$.
}
\label{alg:main}
\end{algorithm}

Given a set of IO examples, Algorithm~\ref{alg:main} uses $A^*$ search over generalized partial programs $\hat{p}\in\hat{\mathcal{P}}$ in conjunction with the heuristic $\hat{p}\mapsto\mu(\den{\hat{p}}_{\phi}^\#(Z))$ to compute the optimal program. In particular, it uses a heap $h$ to keep track of the generalized partial program $\hat{p}$ in the frontier of the search DAG; at each iteration, it pops the current best node $\hat{p}$,
and then enumerates the children $\hat{p}'$ of $\hat{p}$ and adds them to $h$ according to the heuristic. Termination occurs when the maximum over $\hat{p} \in h$ of the objective lower bound $\nu : \hat{\mathbb{R}} \to \mathbb{R}$ (analogous to $\mu$, but lower-bounding rather than upper-bounding the abstract objective) is within a tolerance $\varepsilon \geq 0$ of the maximum over $\hat{p} \in h$ of the upper bound $\mu$, returning the generalized partial program with the highest lower-bound. When when the objective abstraction $\hat{\mathbb{R}}$ is real intervals, $\nu((r,r')) = r$ is a natural choice.

For our running example, consider the following generalized partial programs and their abstract objective values:
\begingroup
\allowdisplaybreaks
\begin{align*}
    \hat{p}_0 &\coloneqq \operatorname{map}(d \leq {??}_{[0,100]}) &&\den{\hat{p}_0}^\#_\phi = [1/2, 1] \\
    \hat{p}_1 &\coloneqq \operatorname{map}(d \leq {??}_{[0,50]}) &&\den{\hat{p}_1}^\#_\phi = [1/2, 1/2] \\
    \hat{p}_2 &\coloneqq \operatorname{map}(d \leq {??}_{[50,100]}) &&\den{\hat{p}_2}^\#_\phi = [1/2, 1] \\
    \hat{p}_3 &\coloneqq \operatorname{map}(d \leq {??}_{[50,75]}) &&\den{\hat{p}_3}^\#_\phi = [1/2, 1] \\
    \hat{p}_4 &\coloneqq \operatorname{map}(d \leq {??}_{[75,100]}) &&\den{\hat{p}_4}^\#_\phi = [1, 1]
\end{align*}
\endgroup
The search process starts with $h_0 = \{(\hat{p}_0 : 1)\}$. In the first iteration, $\hat{p}_0$ is popped, since it (trivially) has the highest heuristic value. Its children, $\hat{p}_1$ and $\hat{p}_2$ are pushed, resulting in $h_1 = \{\hat{p}_2 : 1, \hat{p}_1 : 1/2\}$.

Next, we pop $\hat{p}_2$, since it has the highest heuristic value. Here we can see the ``soft pruning'' of $A^*$ at work: we know descendents of $\hat{p}_1$ cannot achieve objective values $> 1/2$, while---as far as we know---descendents of $\hat{p}_2$ might achieve objective values as high as $1$. We can't prune $\hat{p}_1$ \emph{per se}, since the optimal objective value may be $\leq 1/2$, but we can defer considering descendents of $\hat{p}_1$ until we know that the optimal objective value is $\leq 1/2$. We continue by pushing the children of $\hat{p}_2$, which are $\hat{p}_3$ and $\hat{p}_4$, resulting in $h_2 = \{\hat{p}_3,\ : 1, \hat{p}_4 : 1, \hat{p}_1 : 1/2\}$.

Finally, we observe that, for the abstract objective intervals in the heap ($\{[1/2, 1], [1, 1], [1/2, 1/2]\}$), the distance between the greatest lower bound ($1$) and greatest upper bound ($1$) is $0$, and so search terminates, returning $\hat{p}_4$.

Let $\phi^*_Z \coloneqq \operatorname*{\max}_{p'\in\mathcal{P}}\phi(p',Z)$ be the optimal objective value. We have the following optimality guarantee:
\begin{theorem}
If Algorithm~\ref{alg:main} returns an abstract program $\hat{p}$, then $\hat{p}$ is $\epsilon$-optimal---i.e. $\forall p \in \gamma(\hat{p})$, $\phi^*_Z - \phi(p,Z) \leq \epsilon$.
\end{theorem}
\begin{proof}
By (\ref{eqn:dagcond}), Algorithm~\ref{alg:main} preserves the invariant that every program $p\in\mathcal{P}$ is contained in some generalized partial program $\hat{p}\in h$---i.e.,
\begin{align*}
\bigcup_{\hat{p}\in h}\gamma(\hat{p})=\mathcal{P}.
\end{align*}
This property can be checked by an easy induction argument on the while loop iteration. This implies that
\begin{align*}
    \phi^*_Z &= \max\{\max\{\phi(p',Z) \mid p' \in \gamma(\hat{p}')\} \mid \hat{p}' \in h\}
    \intertext{
        and since $\mu$ must overapproximate the concrete objective values,
    }
    \phi^*_Z &\leq
    \max\{\mu(\den{\hat{p}'}_{\phi}^\#(Z)) \mid \hat{p}' \in h\}.
    \intertext{
        Next, suppose Algorithm~\ref{alg:main} terminates, returning $\hat{p}\in\hat{\mathcal{P}}$, and let $p$ be any concrete program in $\gamma(\hat{p})$. The termination condition ensures that
    }
    \phi^*_Z &\leq \nu(\den{\hat{p}}_{\phi}^\#(Z)) + \epsilon.
    \intertext{
        Finally, since $\nu$ must underapproximate the concrete objective values,
    }
    \phi^*_Z &\leq \phi(p,Z) + \epsilon.
\end{align*}

\end{proof}
We can ensure termination straightforwardly by using a finite DAG as the search space; for instance, we can do so by discretizing the real-valued constants.
In general, we can guarantee convergence when the abstract losses of every infinite chain converge. For example, this property holds when the objective is Lipschitz continuous in the real-valued program parameters. (The gap between the upper and lower bounds of the objective value is bounded by the Lipschitz constant times the diameter of the box, which goes to zero as search proceeds and the boxes become smaller.) However, Lipschitz continuity often does not hold; we leave exploration of alternative ways to ensure convergence to future work.

\section{Instantiation for Interval Domains}
\label{sec:monotone}

Section~\ref{sec:framework} described a general framework for optimal synthesis when given an abstract semantics and search DAG for the target DSL. In this section, we describe a natural strategy for constructing abstract semantics when the
DSL types are partially ordered. We begin by showing how to construct abstract domains for partially ordered types (Section~\ref{sec:intervaldomain}), abstract semantics for individual components with monotone semantics (Section~\ref{sec:monotonetransformer}), and abstract transformers for monotone objectives (Section~\ref{sec:monotoneobjective}). Next, we describe a space of partial programs where holes corresponding to real-valued constants are optionally annotated with interval constraints (Section~\ref{sec:interval:programs}). Finally, we show how to combine %
abstract transformers
and objectives to perform abstract interpretation for our interval-constrained partial programs (Section~\ref{sec:intervalabstractinterpretation}).

\subsection{Interval Domains from Partial Orders}
\label{sec:intervaldomain}

Many DSLs have the property that their types are equipped with a partial order---e.g. the usual order on the real numbers, or the order $\bfalse \leq \btrue$ on the Booleans.

\begin{definition}
\label{defn:intervaldomain}
\rm
Given a partially ordered set $\mathcal{Z}$, let $\bar{\mathcal{Z}}=\mathcal{Z}\cup\{-\infty,+\infty\}$, where $-\infty\le z\le+\infty$ for all $z\in\mathcal{Z}$. Then, the \emph{interval domain} is the set $\hat{\mathcal{Z}}=\{(a,b) \in \bar{\mathcal{Z}}^2 \mid a \leq b\} \cup \{\bot\}$,
together with an abstraction function $\alpha:\mathcal{Z}\to\hat{\mathcal{Z}}$ defined by $\alpha(z)=(z,z)$, and a concretization function $\gamma:\hat{\mathcal{Z}}\to2^{\mathcal{Z}}$ defined by
\begin{align*}
\gamma((z_0,z_1))=\{z\in\mathcal{Z}\mid z_0\le z\le z_1\}.
\end{align*}
\end{definition}
In other words, $\gamma$ maps $(z_0,z_1)$ to the interval $[z_0,z_1]$. It is clear that $z\in\gamma(\alpha(z))$.

\subsection{Interval Transformers for Monotone Functions}
\label{sec:monotonetransformer}

Next, we define our notion of abstract transformer for the interval domain. DSLs often contain functions that are monotone, respecting the partial orders between their input and output types.
\begin{definition}
\rm
Consider a function $f:\mathcal{X}_1\times...\times\mathcal{X}_k\to\mathcal{Y}$, and where $\mathcal{X}_1,...,\mathcal{X}_k$ and $\mathcal{Y}$ are partially ordered sets. We say $f$ is \emph{monotone} if
\begin{align*}
\bigwedge_{i=1}^kx_i\le x_i'\Rightarrow f(x_1,...,x_k)\le f(x_1',...,x_k').
\end{align*}
\end{definition}
For example, the $+$ operator is
monotone in both of its inputs.

\begin{definition}
Given a monotone function $f:\mathcal{X}_1\times...\times\mathcal{X}_k\to\mathcal{Y}$, let $\hat{\mathcal{X}}_i$ be the interval domain for each $\mathcal{X}_i$ and $\hat{\mathcal{Y}}$ be the interval domain for $\mathcal{Y}$. The \emph{interval transformer} for $f$ is the function $\hat{f}:\hat{\mathcal{X}}_1\times...\times\hat{\mathcal{X}}_k\to\hat{\mathcal{Y}}$ defined by
\begin{align*}
\hat{f}((x_{1,0},x_{1,1}),...,(x_{k,0},x_{k,1}))=(f(x_{1,0},...,x_{k,0}),f(x_{k,1},...,x_{k,1})).
\end{align*}
If $x_{i,0}=-\infty$ for any $i\in[k]$, then we let the lower bound be $f(x_{1,0},...,x_{k,0})=-\infty$, and if $x_{i,1}=+\infty$ for any $i\in[k]$, then we let the upper bound be $f(x_{k,1},...,x_{k,1})=+\infty$.\footnote{Here we took monotone to mean monotonically increasing. For functions that are monotonically decreasing, we can flip the endpoints of the interval transformer output.}
\end{definition}

For example, since $+ : \mathbb{R} \times \mathbb{R} \to \mathbb{R}$ is monotone, then $\hat{+} : \hat{\mathbb{R}} \times \hat{\mathbb{R}} \to \hat{\mathbb{R}}$ is given by
\begin{align*}
(a,b)\,\hat{+}\,(c,d) \coloneqq (a + c, b + d)
\end{align*}
The following key result shows that $\hat{f}$ overapproximates the concrete semantics:
\begin{lemma}
\label{lem:intervaltransformer}
Let $f:\mathcal{X}_1\times...\times\mathcal{X}_k\to\mathcal{Y}$ be monotone, and let $\hat{f}$ be its interval transformer. For any $x_i\in\mathcal{X}_i$ and $\hat{x}_i=(x_{i,0},x_{i,1})\in\hat{\mathcal{X}}_i$ such that $x_i\in\gamma(\hat{x}_i)$ for all $i\in[k]$, we have
\begin{align*}
f(x_1,...,x_k)\in\gamma\left(\hat{f}(\hat{x}_1,...,\hat{x}_k)\right).
\end{align*}
\end{lemma}
\begin{proof}
For all $i\in[k]$, by our assumption that $x_i\in\gamma(\hat{x}_i)$, we have $x_{i,0}\le x_i\le x_{i,1}$. Thus, by monotonicity of $f$, we have
\begin{align}
\label{eqn:lem:intervaltransformer}
f(x_{1,0},...,x_{k,0})\le f(x_1,...,x_k)\le f(x_{1,1},...,x_{k,1}).
\end{align}
By definition of $\hat{f}$, we have $\hat{f}(\hat{x}_1,...,\hat{x}_k)=(f(x_{1,0},...,x_{k,0}),f(x_{1,1},...,x_{k,1}))$, so by (\ref{eqn:lem:intervaltransformer}) and definition of $\gamma$, we have $f(x_1,...,x_k)\in\gamma(\hat{f}(\hat{x}_1,...,\hat{x}_k))$, as claimed.
\end{proof}
In other words, if $x_i$ is contained in the interval $\hat{x}_i$ for each $i$, then $f(x_1,...,x_k)$ is contained in the interval $\hat{f}(\hat{x}_1,...,\hat{x}_k)$. Thus, $\hat{f}$ overapproximates the concrete semantics of $f$.

\subsection{Interval Transfomers for Monotone Objectives}
\label{sec:monotoneobjective}

Similarly, if $\phi_0$ is monotone---i.e. for
$W_i=\{(y_{i,1}',y_{i,1}),...,(y_{i,k}',y_{i,k})\}$ if $y_{0,j}'\le y_{1,j}'$ for all $j\in[k]$, then we have
\begin{align*}
\phi_0(W_0)\le\phi_0(W_1).
\end{align*}
Note that we only require monotonicity in the labels $y_{j,i}'$ output by a candidate program $p$, not the ground truth labels $y_{j,i}$, since the latter are always concrete values. Then, we can construct the abstract transformer $\hat{\phi}_0:(\hat{\mathcal{Y}}\times\mathcal{Y})^*\to\hat{\mathbb{R}}$ by
\begin{align*}
\hat{\phi}_0\left(\{((y_{0,i}',y_{1,i}'),y_i)\}_{i=1}^k\right)
=\left(\phi_0\left(\{(y_{0,i}',y_i)\}_{i=1}^k\right),\phi_0\left(\{(y_{1,i}',y_i)\}_{i=1}^k\right)\right).
\end{align*}

\subsection{Partial Programs with Interval Constraints}
\label{sec:interval:programs}
Next, we describe a space of generalized partial programs which extend partial programs with hole annotations that constrain the values that can be used to fill those holes.
\begin{definition}
\rm
Assume that the space of constants $\mathcal{C}$
is partially ordered, and let $\hat{\mathcal{C}}$ be its interval domain. Then, an \emph{interval-constrained partial program} $\tilde{p}=(\hat{p},\kappa)$ is a partial program $\hat{p}$ together with a mapping $\kappa:\text{holes}(\hat{p})\to\hat{\mathcal{C}}\cup\{\varnothing\}$, where $\text{holes}(\hat{p})\subseteq\mathbb{N}$ are the indices of the holes in $\hat{p}$.
\end{definition}
For example, suppose that $\hat{p}=N_1...N_i...N_h$ is a partial program, where $N_i$ is a nonterminal (and therefore a hole), so $i\in\text{holes}(\hat{p})$. Intuitively, an annotation $\kappa(i)=\hat{c}$ imposes the constraint that the value used to fill $N_i$ must be a constant $c\in\mathcal{C}$, and that $c$ must satisfy $c\in\gamma(\hat{c})$ (i.e., $c$ is contained in the interval $[c_0,c_1]$). Alternatively, if $\kappa(i)=\varnothing$, then no such constraint is imposed---i.e., $N_i$ may be filled with any constant or a different production in the grammar.  (Note that $\varnothing$ is not the same as providing the interval $[-\infty,+\infty]$ because that constraint would require this hole to be filled by a \textit{constant}, whereas the optimal program might need a non-constant expression.)
We let $\tilde{\mathcal{P}}$ denote the space of interval-constrained partial programs.

These constraints are imposed by the structure of the search DAG, which is a extended version of the search DAG over partial programs. Below, we first describe the children function for interval-constrained partial programs; the concretization function is constructed from the children function.

\subsubsection*{Children Function.}

Next, we describe the children of an interval constrained partial program $\text{children}(\tilde{p})\subseteq\tilde{\mathcal{P}}$. Intuitively, if a hole $N$ in a partial program $\hat{p}$ can be filled with a constant value $c\in\mathcal{C}$ (i.e., there is a production $N\to c$) to obtain $\hat{p}'$, then $\hat{p}'\in\text{children}(\hat{p})$. In contrast, for an interval-constrained partial program $\tilde{p}$, we include a child annotating $N$ with $[-\infty,\infty]$. Then, subsequent children can further split interval constraints to obtain finer-grained interval constraints (the concrete constant value $c$ can be represented by the constraint $(c,c)\in\hat{\mathcal{C}}$). To formalize this notion, we first separate out productions for constants.
\begin{definition}
\rm
A production $r\in R$ is \emph{constant} if it has the form $r=N\to c$ for some $c\in\mathcal{C}$.
\end{definition}
Now, we can partition $R$ into the set $R_{\mathcal{C}}$ of constant productions and its complement $\hat{R}=R\setminus R_{\mathcal{C}}$, which we call \emph{non-constant productions}. Next, given an interval-constrained partial program $\tilde{p}=(\hat{p},\kappa)$, its holes are simply the holes of the underlying partial program $\hat{p}$: $\text{holes}(\tilde{p})=\text{holes}(\hat{p})$. Then, we partition these holes into \emph{unannotated holes} and \emph{annotated holes}---i.e.,
\begin{align*}
\text{holes}_\varnothing((\hat{p},\kappa))&=\{i\in\text{holes}(\hat{p})\mid\kappa(i)=\varnothing\} \\
\text{holes}_{\mathcal{C}}((\hat{p},\kappa))&=\{i\in\text{holes}(\hat{p})\mid\kappa(i)\neq\varnothing\},
\end{align*}
respectively. Finally, we define three kinds of children for an interval-constrained partial program $\tilde{p}$. First, we include children obtained by filling an unannotated hole with a non-constant production:
\begin{align*}
\text{children}_{\varnothing}(\tilde{p})=\{(\hat{p}',\kappa')\in\tilde{\mathcal{P}}\mid\exists i\in\text{holes}_\varnothing(\tilde{p}),r\in\hat{R}\;.\;\hat{p}=\text{fill}(\hat{p},i,r)\wedge\kappa'=\text{repair}(\kappa;\hat{p},i,r)\}.
\end{align*}
These children are the same as the children constructed in the original search DAG over partial programs. Here, $\text{repair}(\kappa;\hat{p},i,r)$ ``repairs'' $\kappa$ by accounting for how the indices in $\hat{p}$ change after applying production $r$ to fill hole $i$ in $\hat{p}$. In particular, this operation changes the indices of nonterminals in $\hat{p}$; $\kappa'$ accounts for these changes without modifying the annotations themselves. Formally, if $\hat{p}'=\text{fill}(\hat{p},N_i\to M_1...M_h,i)$, with $\hat{p}=N_1...N_i...N_k$ and $\hat{p}'=N_1...M_1...M_h...N_k$, then we have
\begin{align*}
\kappa'(j)=\begin{cases}
\kappa(j)&\text{if }j<i \\
\varnothing&\text{if }i\le j\le h \\
\kappa(j-h+1)&\text{if }j>h.
\end{cases}
\end{align*}
In particular, $\kappa'$ includes the same annotations as $\kappa$.

Second, we consider children obtained by filling an unannotated hole with the interval $[-\infty,\infty]$:
\begin{align*}
\text{children}_{\infty}(\tilde{p})=\{(\hat{p}',\kappa')\in\tilde{\mathcal{P}}\mid\exists i\in\text{holes}_{\varnothing}(\tilde{p})\;.\;\hat{p}'=\hat{p}\wedge\kappa'(i)=[-\infty,\infty]\}.
\end{align*}
In other words, the partial program $\hat{p}$ remains unchanged, but we introduce an annotation onto one of the previously unannotated holes of $\tilde{p}$.

Third, we consider children obtained by replacing an annotation with a tighter annotation:
\begin{align*}
\text{children}_{\mathcal{C}}(\tilde{p})=\{(\hat{p}',\kappa')\in\tilde{\mathcal{P}}\mid\exists i\in\text{holes}_{\mathcal{C}}(\tilde{p})\;.\;\hat{p}'=\hat{p}\wedge\text{subset}(\kappa'(i),\kappa(i))\}.
\end{align*}
Here, $\text{subset}(\hat{c},\hat{c}')$ checks whether $\hat{c}=(c_0,c_1)$ is a \emph{strict} subset of $\hat{c}'=(c_0',c_1')$, that is, $c_0'\le c_0$ and $c_1\le c_1'$, and one of these inequalities is strict; equivalently, $[c_0,c_1]\subsetneq[c_0',c_1']$.

Finally, our overall search DAG is defined by the union of these children:
\begin{align}
\label{eqn:children}
\text{children}(\tilde{p})=\text{children}_{\varnothing}(\tilde{p})\cup\text{children}_{\infty}(\tilde{p})\cup\text{children}_{\mathcal{C}}(\tilde{p}).
\end{align}
Note that by defining children in this way, there may be infinitely many children; in addition, multiple children may cover the same concrete program, leading to redundancy in the search DAG. Practical implementations can restrict to a finite subset of these children as long as they satisfy (\ref{eqn:dagcond})---i.e., the union of the concrete programs in the children of $\tilde{p}$ cover all concrete programs in $\tilde{p}$. In addition, these children are ideally chosen so the overlap is minimal.

\subsubsection*{Concretization Function.}

Recall that the concretization function $\gamma(\tilde{p})$ contains concrete programs $p$ represented by $\tilde{p}$. We take the approach where we define $\gamma$ based on the children function---i.e., $p\in\gamma(\tilde{p})$ if there exists a sequence $\tilde{p}_1,...,\tilde{p}_n=p$ such that $\tilde{p}_1 = \tilde{p}$, $\tilde{p}_n = p$, and $\tilde{p}_{j+1}\in\text{children}(\tilde{p}_j)$ for all $1<j<n$.

\subsection{Interval Transformers for Partial Programs with Interval Constraints}
\label{sec:intervalabstractinterpretation}

Next, we describe how to implement abstract interpretation for partial programs
with interval constraints.
While abstract interpretation is typically performed with respect to program inputs, in our case it is with respect to program constants.
We assume all components $f\in\mathcal{F}$ have an abstract transformer $\hat{f}$, and the objective
$\phi_0$ has an abstract transformer $\hat{\phi}_0$ (if they are monotone, their abstract transformers can be defined as in Sections~\ref{sec:monotonetransformer} and~\ref{sec:monotoneobjective}).

First, we modify
the grammar of programs so that the constants $\mathcal{C}$ are replaced by
abstract values $(c_0,c_1)\in\hat{\mathcal{C}}$. While the concrete semantics cannot be applied to these programs, we will define abstract semantics for them. Now, given a generalized partial program $\tilde{p}=(\hat{p},\kappa)$, for each unannotated hole $i\in\text{holes}_{\varnothing}(\tilde{p})$, we replace the corresponding nonterminal $N_i$ in $\hat{p}$ with the abstract value $(-\infty,\infty)$, and for each annotated hole $i\in\text{holes}_{\mathcal{C}}(\tilde{p})$, we replace the corresponding nonterminal $N_i$ with the annotation $\kappa(i)$. Finally, we replace any constant $c$ in $\hat{p}$ with the abstract value $(c,c)$. Once we have performed this transformation, we can define the following abstract semantics for $\hat{p}$:
\begin{align*}
\den{X}^\#(x)=\alpha(x),\quad
\den{\hat{c}}^\#(x)=\hat{c},\quad
\den{f(p_1,...,p_k)}^\#(x)=\hat{f}(\den{p_1}^\#(x),...,\den{p_k}^\#(x)),
\end{align*}

Now, we can combine $\den{\cdot}^\#$ with $\hat{\phi}_0$ to obtain the abstract objective $\den{\cdot}_{\phi}^\#$:
\begin{align*}
\den{(\hat{p},\kappa)}_{\phi}^\#(Z)
=\hat{\phi}_0\left(\{(\den{\hat{p}}^\#(x),y)\}_{(x,y)\in Z}\right).
\end{align*}
In other words, we apply the abstract semantics $\den{\hat{p}}^\#$ to each input $x$, obtain the corresponding abstract output $\hat{y}'$, and apply $\hat{\phi}_0$ to the resulting set $\{(\hat{y}',y)\}$.

\section{Implementation}
\label{sec:impl}

We instantiate our framework for two different domain specific languages (DSLs):
\begin{itemize}
\item \textbf{NEAR DSL:} The NEAR language for the CRIM dataset (Section~\ref{sec:crim}). The motivating example described in Section~\ref{sec:motivating} is derived from this setting.
\item \textbf{Quivr DSL:} A query language over video trajectories, which uses constructs similar to regular expressions to match trajectories (Section~\ref{sec:quivr}).
\end{itemize}
Although both of these DSLs process sequence data, their computation models are quite different: NEAR focuses on folding combinators over the inputs, whereas Quivr's primary operations reduce to matrix multiplication.

For both DSLs, we refine the definition of children compared to Equation~\ref{eqn:children}. In particular, in Section~\ref{sec:interval:programs}, when defining the search space over programs with interval constraints on holes, $\text{children}_{\mathcal{C}}$ is defined such that the children of $[a,b]$ are \textit{all} of its strict subintervals. As discussed there, this means that each node may have
infinitely many children, which is impractical for an implementation. Instead, our implementation splits intervals into just two children---i.e., the children of $[a,b]$ are $[a,(a+b)/2]$ and $[(a+b)/2,b]$. These intervals partition $[a,b]$, so all concrete programs are still contained in the search space, retaining soundness. This splits an interval into child intervals of equal length, but with additional domain-knowledge, other choices could be made (e.g. with a prior distribution over parameters, splitting into intervals of equal probability mass).

Finally, we also describe how we construct the abstract transformer for the $F_1$ score, which is commonly used as the objective function in binary classification problems (Section~\ref{sec:absobj}).

\subsection{NEAR DSL for Trajectory Labeling}
\label{sec:crim}

In the NEAR DSL~\cite{shah2020learning}, the program input is a featurized trajectory, which is a sequence of feature vectors $x\in(\mathbb{R}^n)^*$, where
$n$ is the dimension of the feature vector for each frame. The output is a sequence of labels $y\in\{\btrue,\bfalse\}^*$ of the same length as the input, where $y[t]=\btrue$ if a frame exhibits the given behavior and $y[t]=\bfalse$ otherwise (i.e. the task is binary classification at the frame level).

\subsubsection*{Syntax.}

This DSL has three kinds of expressions, encoded by their corresponding nonterminal in $\mathcal{G}$: (i) $vv$ represents functions mapping feature vectors to real values (e.g.
the body of map), (ii) $\ell v$ represents functions mapping lists to real values (e.g. fold), and (iii) $\ell\ell$ represents functions mapping lists to lists (e.g. map). In particular, this DSL has the following productions:
\begin{align*}
vv &\Coloneqq z_i \mid z_f \mid c \in \mathbb{R} \mid vv + vv \mid vv \cdot vv \mid \kw{ite}(vv, vv, vv) \qquad (i\in[n])\\
\ell v &\Coloneqq \kw{fold}(vv) \mid \kw{ite}(\ell v, \ell v, \ell v) \\
\ell\ell &\Coloneqq \kw{map}(vv) \mid \kw{mapprefix}(\ell v) \mid \kw{ite}(\ell v, \ell\ell, \ell\ell),
\end{align*}
The DSL syntax is in the combinatory style, so $\lambda$s are omitted. In particular, the expressions encoded by $vv$ are combinators designed to be used within a higher-order function such as \kw{map} or \kw{fold}. These combinators are applied to individual elements of the input list, where $z_i$ is a variable representing the $i$th element of the current feature vector $z=x[t]$, and $z_f$ is a special symbol used inside \kw{fold} to represent its accumulated running state.

The start symbol is $\ell\ell$.
The DSL is structured such that a list of feature vectors $x$ is mapped to a list of real values $r$ of the same length. Each real-valued output $r[t]$ implicitly encodes the label
\begin{align*}
y[t]=
\begin{cases}
\btrue&\text{if }r[t]\ge0 \\
\bfalse&\text{otherwise}.
\end{cases}
\end{align*}

The running example involved programs in a toy DSL of the form $\operatorname{map}(d \leq c)$. In the NEAR DSL, this would be represented as $\operatorname{map}(-1 \cdot z_4 + c)$, as ``distance between mice'' is the $4$th feature, and the label will be obtained by comparing the program output to $0$.

\subsubsection*{Semantics.}
We let $V = \mathbb{R}^n$ (where $n$ is the dimension of the feature space) be the space of feature vectors, $L = V^*$
is the space of trajectories, and
$K = \mathbb{R}^*$ is the space of output sequences. Then, the nonterminal $vv$ in this DSL has semantics $\den{vv} : V \times \mathbb{R} \to \mathbb{R}$, where
\begingroup
\allowdisplaybreaks
\begin{alignat*}{2}
&\den{z_i}(v, s) &&\coloneqq v_i \\
&\den{z_f}(v, s) &&\coloneqq s \\
&\den{c}(v, s) &&\coloneqq c \\
&\den{vv_1 + vv_2}(v, s) &&\coloneqq \den{vv_1}(v, s) + \den{vv_2}(v, s) \\
&\den{vv_1 \cdot vv_2}(v, s) &&\coloneqq \den{vv_1}(v, s) \cdot \den{vv_2}(v, s) \\
&\den{\kw{ite}(vv_1,vv_2,vv_3)}(v, s) &&\coloneqq \kw{if}\;(\den{vv_1}(v, s) \geq 0)\;\kw{then}\; \den{vv_2}(v, s)\;\kw{else}\;\den{vv_3}(v, s).
\intertext{
    Next, the nonterminal $\ell v$ in this DSL has semantics $\den{\ell v} : V^* \to \mathbb{R}$, where
}
&\den{\kw{fold}(vv)}(\ell) &&\coloneqq \kw{fold}(\lambda v \lambda s.\; \den{vv}(v,s),\ell,0)\\
&\den{\kw{ite}(\ell v_1,\ell v_2,\ell v_3)}(\ell) &&\coloneqq \kw{if}\;(\den{\ell v_1}(\ell) \geq 0)\;\kw{then}\; \den{\ell v_2}(\ell)\;\kw{else}\; \den{\ell v_3}(\ell).
\intertext{
Here, $\kw{fold} : (V \to \mathbb{R} \to \mathbb{R}) \to V^* \to \mathbb{R} \to \mathbb{R}$ is standard, and passes the intermediate value as the $s$ argument to $vv$.
The nonterminal $\ell\ell$ in this DSL has semantics $\den{\ell\ell} : L \to K$, where
}
&\den{\kw{map}(vv)}(\ell) &&\coloneqq \kw{map}(\lambda v.\;\den{vv}(v,0),\ell) \\
&\den{\kw{mapprefixes}(\ell v)}(\ell) &&\coloneqq \kw{map}(\lambda \ell'.\;\den{\ell v}(\ell'),\kw{prefixes}(\ell))\\
&\den{\kw{ite}(\ell v,\ell\ell_1,\ell\ell_2)}(\ell) &&\coloneqq \kw{if}\; (\den{\ell v}(\ell) \geq 0)\; \kw{then}\; \den{\ell\ell_1}(\ell)\;\kw{else}\; \den{\ell\ell_2}(\ell).
\end{alignat*}
\endgroup

Here, $\kw{map}$ is standard\footnote{But note that $z_f$ is treated as 0 if it appears inside an expression not within \kw{fold}} and
\begin{align*}
\kw{prefixes}:
(x_1, \ldots, x_n) \mapsto ((x_1), (x_1, x_2), (x_1, x_2, x_3), \ldots (x_1, \ldots, x_n)).
\end{align*}
Finally, as described above, the labels $y$ are obtained by thresholding $\den{\ell\ell}(\ell)$. That is, we let $\den{\ell\ell}^b : L \to \{\btrue,\bfalse\}^*$ denote the label
\begin{align*}
\pparen{\den{\ell\ell}^b(\ell)}[t] =
\begin{cases}
\btrue &\text{if }\pparen{\den{\ell\ell}(\ell)}[t] \geq 0 \\
\bfalse &\text{otherwise}.
\end{cases}
\end{align*}

\subsubsection*{Abstract Semantics.}

We abstract $\mathbb{R}$ with the usual real intervals, $\hat{\mathbb{R}}$. We abstract $K$ with $(\hat{\mathbb{R}})^*$, products of real intervals.
Addition is monotone, and multiplication $[a_1,b_1] \cdot [a_2, b_2]$ can be shown to be abstracted by $[\min(a_1 b_1, a_1 b_2, a_2 b_1, a_2 b_2), \max(a_1 b_1, a_1 b_2, a_2 b_1, a_2 b_2)]$.
We can represent $\kw{ite}(c,a,b)$ as $\indic{c \geq 0} \cdot a + \indic{c < 0} \cdot b$, and thus can abstract %
it using addition, multiplication, and the abstraction sending $\indic{[\bfalse, \bfalse]} = [0, 0]$, $\indic{[\bfalse, \btrue]} = [0,1]$, and $\indic{[\btrue, \btrue]} = [1,1]$.

\subsubsection*{Search Space.}

The search space over $vv$ has a redundancy due to commutativity, associativity, and distributativity, which unduly hinders search. Instead, we use a \textit{normalized} version, where $vv$ expressions are constrainted to be sums-of-products (fully distributed), and where we ignore commutativity and associative in the sums of products. Essentially, we only consider polynomials, where the variables are the features $z_i$, the fold variable $z_f$, and indicator variables $\indic{vv \geq 0}$ for each $vv$ (to maintain the expressiveness of \kw{ite}).
Further, we consider only constants in $[-1, 1]$ and we also
normalize the dataset so each feature $z_i$ is in $[-1, 1]$.

We define a notion of size for programs, so that we can bound the search space, where \kw{ite}, \kw{map} and \kw{mapprefix} have size $1$, \kw{fold} has size $0$ (since \kw{mapprefix} must contain a \kw{fold} and already has size 1). Each monomial in the polynomial has size $1$, and each polynomial variable in a monomial has size $1$. In a given polynomial, the size of indicators is amortized: each nested $vv$ has size $1$ and produces a new polynomial variable.

\subsection{Quivr DSL For Trajectory Queries}
\label{sec:quivr}

In the Quivr DSL~\cite{mell2023synthesizing}, the program input is a featurized trajectory, which is a sequence $x\in(\mathbb{R}^n)^*$ as before. However, the output is now a single label $y\in\{\btrue,\bfalse\}$ for the entire trajectory. This DSL is designed to allow the user to select trajectories that satisfy certain properties---e.g. the user may want to identify all cars that make a right turn at a certain intersection in a traffic video.

\subsubsection*{Syntax.}

This DSL is based on the Kleene algebra with tests~\cite{kozen1997kleene}, which, intuitively, are regular expressions where the ``characters'' are actually predicates. Its syntax is
\begin{alignat*}{2}
Q &\Coloneqq f \mid g(C) \mid Q \seq Q \mid Q \land Q \qquad &&(f\in\mathcal{F}_{\varnothing}, g\in\mathcal{F}_{\mathcal{C}}) \\
C &\Coloneqq c \qquad &&(c\in\mathbb{R})
\end{alignat*}
where $\mathcal{F}_{\varnothing}$ and $\mathcal{F}_{\mathcal{C}}$ are given sets of domain-specific predicates, the latter of which have constants $c \in \mathbb{R}$ that need to be chosen by the synthesizer.

\subsubsection*{Semantics.}

Expressions in this DSL denote functions mapping sequences of feature vectors $x\in(\mathbb{R}^n)^*$ to whole-sequence labels ($\{\btrue,\bfalse\}$), defined as
\begingroup
\allowdisplaybreaks
\begin{alignat*}{2}
&\den{f}(x) &&\coloneqq f(x) \\
&\den{g(c)}(x) &&\coloneqq g(x)\ge c \\
&\den{Q_1 \land Q_2}(x) &&\coloneqq \den{Q_1}(x) \land \den{Q_2}(x) \\
&\den{Q_1 \seq Q_2}(x) &&\coloneqq \bigvee_{k=0}^n ~ \den{Q_1}(x_{0:k}) \land \den{Q_2}(x_{k:n})
\end{alignat*}
\endgroup
Each predicate $f\in\mathcal{F}_{\varnothing}$ has type $f:(\mathbb{R}^n)^*\to\{\btrue,\bfalse\}$, and indicates whether $x$ matches $f$, and each predicate $g\in\mathcal{F}_{\mathcal{C}}$ has type $g : (\mathbb{R}^n)^*\to \mathbb{R}$, and produces a real-valued score thresholded at a given constant $c$ to indicate whether $x$ matches $g$.

\subsubsection*{Abstract Semantics.}

Note that under the standard orders for reals $\mathbb{R}$ and Booleans $\mathbb{B}=\{\btrue,\bfalse\}$ (i.e., $\bfalse < \btrue$), the semantics $\den{g(c)}(x)$ is monotone decreasing with respect to $c$. Furthermore, both conjunction and disjunction are monotone increasing in their inputs. Thus, we can use our interval transformers for monotone functions from Section~\ref{sec:intervalabstractinterpretation}. In particular, recall that $\hat{\mathbb{R}}$ is the abstract domain of real intervals, and let $\hat{\mathbb{B}}$ be the abstract domain of Boolean intervals 
\begin{align*}
\hat{\mathbb{B}}=\{(\bfalse,\bfalse),(\bfalse,\btrue),(\btrue,\btrue)\}.
\end{align*}
Then, we use the abstract transformers
\begin{align*}
g(x)\ge(c_0,c_1)&=\begin{cases}
(\bfalse,\bfalse)&\text{if }g(x)<c_0 \\
(\bfalse,\btrue)&\text{if }c_0\le g(x)<c_1 \\
(\btrue,\btrue)&\text{if }c_1\le g(x)
\end{cases} \\
\hat{x}_1 \wedge \hat{x}_2&=
\begin{cases}
(\btrue,\btrue) &\text{if }\hat{x}_1=(\btrue,\btrue)\wedge\hat{x}_2=(\btrue,\btrue) \\
(\bfalse,\bfalse) &\text{if }\hat{x}_1=(\bfalse,\bfalse)\vee\hat{x}_2=(\bfalse,\bfalse) \\
(\bfalse,\btrue) &\text{otherwise}
\end{cases} \\
\hat{x}_1 \vee \hat{x}_2&=
\begin{cases}
(\btrue,\btrue) &\text{if }\hat{x}_1=(\btrue,\btrue)\vee\hat{x}_2=(\btrue,\btrue) \\
(\bfalse,\bfalse) &\text{if }\hat{x}_1=(\bfalse,\bfalse)\wedge\hat{x}_2=(\bfalse,\bfalse) \\
(\bfalse,\btrue) &\text{otherwise}.
\end{cases}
\end{align*}
Since the concrete semantics are defined in terms of these operators, composing their abstract versions gives an abstract semantics.

\subsection{Abstract $F_1$ Score}
\label{sec:absobj}

In Section~\ref{sec:monotone}, we described how to construct abstract transformers for monotone functions, which covers most components in these DSLs. However, many objectives commonly used in practice, such as the $F_1$ score, are non-monotone. For non-monotone objectives, we need to provide a custom abstract transformer that overapproximates their concrete semantics for the interval domain.
We now describe how to do so for the $F_1$ score.

Let $W$ be the multiset of outcomes of the form $(y',y)$ where $y'$ is the prediction and $y$ is the ground truth.  Then the $F_1$ score is given by:
\begin{align*}
TP(W) &\coloneqq \sum_{(y', y) \in W} \indic{y = \btrue \wedge y' = \btrue} \\
FP(W) &\coloneqq \sum_{(y', y) \in W} \indic{y = \bfalse \wedge y' = \btrue} \\
F_1(W) &\coloneqq 2 \cdot \frac{TP(W)}{TP(W) + FP(W) + \abs{W^+}},
\end{align*}
where $\abs{W^+}=|\{(y', y) \in W \mid y = \btrue\}|$, $TP$ is the number of true positives, and $FP$ is the number of false positives. Note that $TP$ and $FP$ are monotone in $y'$, so we can use the corresponding interval transformers:
\begin{align*}
TP^\#(\hat{W}) &\coloneqq \sum_{(\hat{y}', y) \in \hat{W}}
\begin{cases}
[1, 1] &\text{if }y = \btrue \wedge \hat{y}' = (\btrue, \btrue) \\
[0, 1] &\text{if }y = \btrue \wedge \hat{y}' = (\bfalse, \btrue) \\
[0, 0] &\text{if }y = \btrue \wedge \hat{y}' = (\bfalse, \bfalse)
\end{cases} \\
FP^\#(\hat{W}) &\coloneqq \sum_{(\hat{y}', y) \in \hat{W}}
\begin{cases}
[1, 1] &\text{if }y = \bfalse \wedge \hat{y}' = (\btrue, \btrue) \\
[0, 1] &\text{if }y = \bfalse \wedge \hat{y}' = (\bfalse, \btrue) \\
[0, 0] &\text{if }y = \bfalse \wedge \hat{y}' = (\bfalse, \bfalse).
\end{cases}
\end{align*}
Letting $TP^\#(W)=[a_1,b_1]$ and $FP^\#(W)=[a_2,b_2]$, and noting that for positive numbers we can abstract division as $[a_1,b_1]/[a_2,b_2]=\left[\frac{a_1}{b_2},\frac{b_1}{a_2}\right]$, a na\"{i}ve strategy of applying abstractions for $+$, $\cdot$, and $/$ leads to the abstract $F_1$ score
\begin{align*}
F_1^\#(W) &\coloneqq 2 \cdot \frac{[a_1, b_1]}{[a_1, b_1] + [a_2, b_2] + \abs{W^+}}
= 2 \cdot \left[\frac{a_1}{b_1 + b_2 + \abs{W^+}},\frac{b_1}{a_1 + a_2 + \abs{W^+}}\right],
\end{align*}
where $\abs{W^+}$ is independent of $y'$, so we can treat it as a constant. 
However, this abstraction is very loose, and we found it not to be useful in practice---e.g. it can be as loose as $[0,2]$ even though $F_1$ scores never exceed $1$.
Instead, we can rewrite
\begin{align*}
F_1(W) 
= 2 \cdot \frac{\frac{TP(W)}{FP(W) + \abs{W^+}}}{\frac{TP(W)}{FP(W) + \abs{W^+}} + 1}
= 2 \cdot L\pparen*{\frac{TP(W)}{FP(W) + \abs{W^+}}}
\end{align*}
where $L(x) = \frac{x}{x + 1}$ is monotone, which leads to the abstract $F_1$ score
\begin{align*}
F_1^\#(W) \coloneqq&\ 2 \cdot L\pparen*{\frac{[a_1, b_1]}{[a_2, b_2] + \abs{W^+}}}
= 2 \cdot L\pparen*{\left[\frac{a_1}{b_2 + \abs{W^+}}, \frac{b_1}{a_2 + \abs{W^+}}\right]} \\
=&\ 2 \cdot \left[\frac{a_1}{a_1 + b_2 + \abs{W^+}}, \frac{b_1}{b_1 + a_2 + \abs{W^+}}\right].
\end{align*}

\section{Experiments}
\label{sec:exp}

We experimentally evaluate our approach in the context of the NEAR and Quivr DSLs described in Section~\ref{sec:impl}.
We demonstrate that our synthesizer outperforms two synthesis baselines in scalability of running time: Metasketches~\cite{bornholt2016optimizing}, an optimal synthesizer based on SMT solvers (Section~\ref{sec:smt}), as well as an ablation that uses breadth-first search instead of our abstract interpretation based heuristic and $A^*$ search (Section~\ref{sec:bfs}).

\subsection{Experimental Setup}

\subsubsection*{Benchmarks.}

We consider two different neurosymbolic program synthesis benchmarks, based on the DSLs described in Section~\ref{sec:impl}:
\begin{itemize}
\item \textbf{CRIM13:} The NEAR~\cite{shah2020learning} DSL applied to two tasks in the CRIM13 dataset~\cite{crim13} ``sniff'' (A) and ``other'' (B). This dataset consists of featurized videos of two mice interacting in an enclosure, with 12,404 training examples with 100 frames each.
\item \textbf{Quivr:} The Quivr~\cite{mell2023synthesizing} DSL, applied to the 17 tasks that they evaluate on. Of these, 6 tasks are on the MABe22~\cite{mabe22} dataset, a dataset of interactions between 3 mice, and 10 tasks are on the YTStreams~\cite{bastani2020miris} dataset, a dataset of video from fixed-position traffic cameras.
\end{itemize}

\subsubsection*{Compute.}

We ran all experiments on a Intel Xeon Gold 6342 CPU (2.80GHz, 36 cores/72 threads). Our implementation uses PyTorch, a library which provides fast matrix operations.

\subsection{Comparison to Metasketches}
\label{sec:smt}

Metasketches performs optimal synthesis using an SMT solver by, in addition to the correctness specification, adding an SMT constraint that the program's score be greater than the best score seen so far. Thus if the SMT solver returns ``SAT'', a better program will have been found, and the process is repeated. Note that in our setting, there is no correctness specification, and so achieving a better score is the only SMT constraint.

A similar strategy, which we found to be more effective, is to do binary search on the objective score: supposing that the objective is in $[0, 1]$, we ask the SMT solver whether there is a program achieving score at least $1/2$; if it returns ``SAT'', we ask for $3/4$, and if ``UNSAT'' we ask for $1/4$, and so on.
Our implementation uses the Z3~\citep{z3} SMT solver. For a fairer comparison, in this experiment we restricted PyTorch to a single CPU core.

The two approaches were both run until they had converged to the exact optimal program. Figure~\ref{fig:exp:smt} shows that the SMT solver scales very poorly as a function of the number of trajectories in the training dataset. While competitive for three or four trajectories, we would like to evaluate on datasets of hundreds or thousands of trajectories.

\begin{figure}[t]
\centering
\begin{subfigure}[b]{0.49\textwidth}
\centering
\includegraphics[width=\textwidth]{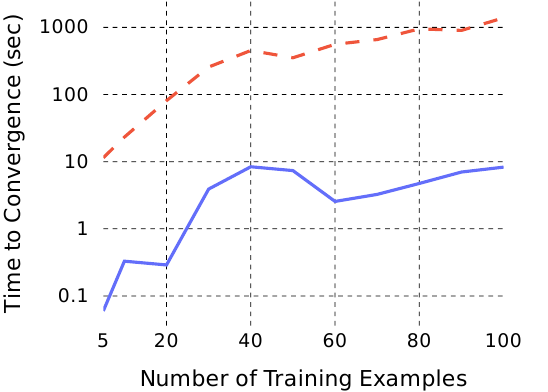}
\vspace{-1em}
\caption{Quivr, task G}
\label{fig:exp:smt:quivr}
\end{subfigure}
\begin{subfigure}[b]{0.49\textwidth}
\centering
\includegraphics[width=\textwidth]{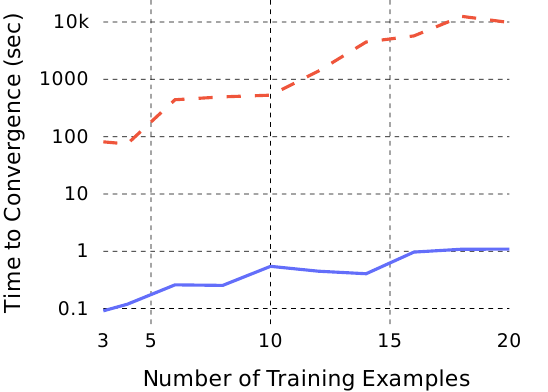}
\vspace{-1em}
\caption{CRIM13, task A}
\label{fig:exp:smt:crim}
\end{subfigure}
\caption{The time (seconds, log scale) to identify the optimal program and prove its optimality, for our approach (\textcolor[HTML]{636efa}{blue}, solid) and an SMT solver (\textcolor[HTML]{ef553b}{red}, dashed), as a function of the size of the training dataset, for two different tasks.}
\label{fig:exp:smt}
\end{figure}

\subsection{Comparison to Breadth-First Search}
\label{sec:bfs}

\begin{table}[t]
    \centering
    \small
    \caption{Best $F_1$ score found and range between upper and lower bounds, at a particular time during search, for different tasks and different algorithms. ``CA'' and ``CB'' are the NEAR CRIM13 queries, and ``QA'' through ``QQ'' are the Quivr queries. ``H'' is $A^*$ search and ``B'' is breadth-first search. For each task and each time, the best $F_1$ score is bolded and the smallest range is bolded. A range of $0$ implies that search has converged to the optimal $F_1$ score.}
    \begin{tabular}{lrrrrrrrrrrrrr}
    \toprule
    \multicolumn{2}{c}{\textbf{Setting}} & \multicolumn{2}{c}{\textbf{10\,sec}} & \multicolumn{2}{c}{\textbf{30\,sec}} & \multicolumn{2}{c}{\textbf{1\,min}} & \multicolumn{2}{c}{\textbf{2\,min}} & \multicolumn{2}{c}{\textbf{5\,min}} & \multicolumn{2}{c}{\textbf{10\,min}} \\
    \midrule
    CA & H & 0.21 & (0.79) & 0.21 & (0.79) & 0.21 & (0.79) & \textbf{0.52} & \textbf{(0.14)} & \textbf{0.53} & \textbf{(0.03)} & \textbf{0.53} & \textbf{(0.00)} \\
     & B & 0.21 & (0.79) & 0.21 & (0.79) & 0.21 & (0.79) & 0.22 & (0.78) & 0.46 & (0.54) & 0.50 & (0.50) \\
    CB & H & 0.75 & (0.25) & 0.75 & (0.25) & 0.75 & (0.25) & \textbf{0.80} & \textbf{(0.05)} & \textbf{0.80} & \textbf{(0.04)} & \textbf{0.80} & \textbf{(0.01)} \\
     & B & 0.75 & (0.25) & 0.75 & (0.25) & 0.75 & (0.25) & 0.75 & (0.25) & 0.75 & (0.25) & 0.75 & (0.25) \\
    QA & H & 0.12 & (0.88) & \textbf{0.77} & \textbf{(0.23)} & \textbf{0.77} & \textbf{(0.23)} & \textbf{0.77} & \textbf{(0.23)} & \textbf{0.77} & \textbf{(0.23)} & \textbf{0.77} & \textbf{(0.23)} \\
     & B & 0.12 & (0.88) & 0.26 & (0.74) & 0.26 & (0.74) & 0.26 & (0.74) & 0.26 & (0.74) & 0.28 & (0.72) \\
    QB & H & 0.27 & (0.73) & 0.31 & (0.69) & \textbf{0.52} & \textbf{(0.48)} & \textbf{0.52} & \textbf{(0.48)} & \textbf{0.52} & \textbf{(0.48)} & \textbf{0.52} & \textbf{(0.48)} \\
     & B & 0.27 & (0.73) & \textbf{0.40} & \textbf{(0.60)} & 0.40 & (0.60) & 0.40 & (0.60) & 0.42 & (0.58) & 0.50 & (0.50) \\
    QC & H & 0.06 & (0.94) & 0.14 & (0.86) & \textbf{0.30} & \textbf{(0.70)} & \textbf{0.33} & \textbf{(0.67)} & \textbf{0.38} & \textbf{(0.62)} & \textbf{0.38} & \textbf{(0.62)} \\
     & B & 0.06 & (0.94) & \textbf{0.25} & \textbf{(0.75)} & 0.25 & (0.75) & 0.26 & (0.74) & 0.26 & (0.74) & 0.26 & (0.74) \\
    QD & H & 0.04 & (0.96) & \textbf{0.25} & \textbf{(0.75)} & \textbf{0.25} & \textbf{(0.75)} & \textbf{0.25} & \textbf{(0.75)} & \textbf{0.40} & \textbf{(0.60)} & \textbf{0.44} & \textbf{(0.56)} \\
     & B & 0.04 & (0.96) & 0.06 & (0.94) & 0.06 & (0.94) & 0.06 & (0.94) & 0.09 & (0.91) & 0.19 & (0.81) \\
    QE & H & 0.04 & (0.96) & \textbf{0.20} & \textbf{(0.80)} & \textbf{0.44} & \textbf{(0.56)} & \textbf{0.44} & \textbf{(0.56)} & \textbf{0.44} & \textbf{(0.56)} & \textbf{0.44} & \textbf{(0.56)} \\
     & B & 0.04 & (0.96) & 0.06 & (0.94) & 0.06 & (0.94) & 0.10 & (0.90) & 0.10 & (0.90) & 0.15 & (0.85) \\
    QF & H & 0.38 & (0.62) & \textbf{0.78} & \textbf{(0.22)} & \textbf{0.78} & \textbf{(0.22)} & \textbf{0.78} & \textbf{(0.22)} & \textbf{0.78} & \textbf{(0.22)} & \textbf{0.78} & \textbf{(0.22)} \\
     & B & 0.38 & (0.62) & 0.42 & (0.58) & 0.42 & (0.58) & 0.55 & (0.45) & 0.56 & (0.44) & 0.60 & (0.40) \\
    QG & H & 1.00 & (0.00) & 1.00 & (0.00) & 1.00 & (0.00) & 1.00 & (0.00) & 1.00 & (0.00) & 1.00 & (0.00) \\
     & B & 1.00 & (0.00) & 1.00 & (0.00) & 1.00 & (0.00) & 1.00 & (0.00) & 1.00 & (0.00) & 1.00 & (0.00) \\
    QH & H & 0.74 & (0.26) & 0.74 & (0.26) & 0.74 & (0.26) & 0.74 & (0.26) & 0.74 & (0.26) & 0.74 & (0.26) \\
     & B & \textbf{0.78} & \textbf{(0.22)} & \textbf{0.78} & \textbf{(0.22)} & \textbf{0.78} & \textbf{(0.22)} & \textbf{0.78} & \textbf{(0.22)} & \textbf{0.78} & \textbf{(0.22)} & \textbf{0.78} & \textbf{(0.22)} \\
    QI & H & 0.75 & (0.25) & \textbf{0.80} & \textbf{(0.20)} & \textbf{0.80} & \textbf{(0.20)} & \textbf{0.80} & \textbf{(0.20)} & \textbf{0.80} & \textbf{(0.20)} & \textbf{0.80} & \textbf{(0.20)} \\
     & B & 0.75 & (0.25) & 0.75 & (0.25) & 0.75 & (0.25) & 0.75 & (0.25) & 0.75 & (0.25) & 0.75 & (0.25) \\
    QJ & H & 0.64 & (0.36) & 0.64 & (0.36) & 0.64 & (0.36) & 0.64 & (0.36) & 0.64 & (0.36) & 0.64 & (0.36) \\
     & B & 0.64 & (0.36) & \textbf{0.73} & \textbf{(0.27)} & \textbf{0.73} & \textbf{(0.27)} & \textbf{0.73} & \textbf{(0.27)} & \textbf{0.73} & \textbf{(0.27)} & \textbf{0.73} & \textbf{(0.27)} \\
    QK & H & 0.75 & (0.25) & 0.75 & (0.25) & 0.75 & (0.25) & 0.75 & (0.25) & 0.75 & (0.25) & 0.75 & (0.25) \\
     & B & 0.75 & (0.25) & 0.75 & (0.25) & 0.75 & (0.25) & 0.75 & (0.25) & 0.75 & (0.25) & 0.75 & (0.25) \\
    QL & H & 0.71 & (0.29) & 0.91 & (0.09) & 0.91 & (0.09) & 0.91 & (0.09) & 0.91 & (0.09) & 0.91 & (0.09) \\
     & B & 0.71 & (0.29) & 0.91 & (0.09) & 0.91 & (0.09) & 0.91 & (0.09) & 0.91 & (0.09) & 0.91 & (0.09) \\
    QM & H & 0.73 & (0.27) & \textbf{0.92} & \textbf{(0.08)} & \textbf{0.92} & \textbf{(0.08)} & \textbf{0.92} & \textbf{(0.08)} & 0.92 & (0.08) & 0.92 & (0.08) \\
     & B & 0.73 & (0.27) & 0.80 & (0.20) & 0.80 & (0.20) & 0.80 & (0.20) & 0.92 & (0.08) & 0.92 & (0.08) \\
    QN & H & 0.73 & (0.27) & 0.73 & (0.27) & \textbf{0.89} & \textbf{(0.11)} & \textbf{0.89} & \textbf{(0.11)} & \textbf{0.89} & \textbf{(0.11)} & \textbf{0.89} & \textbf{(0.11)} \\
     & B & 0.73 & (0.27) & 0.73 & (0.27) & 0.80 & (0.20) & 0.80 & (0.20) & 0.80 & (0.20) & 0.80 & (0.20) \\
    QO & H & 0.67 & (0.33) & 0.67 & (0.33) & \textbf{0.86} & \textbf{(0.14)} & \textbf{0.86} & \textbf{(0.14)} & \textbf{0.86} & \textbf{(0.14)} & \textbf{0.86} & \textbf{(0.14)} \\
     & B & 0.67 & (0.33) & 0.67 & (0.33) & 0.80 & (0.20) & 0.80 & (0.20) & 0.80 & (0.20) & 0.80 & (0.20) \\
    QP & H & 0.50 & (0.50) & 0.50 & (0.50) & 1.00 & (0.00) & 1.00 & (0.00) & 1.00 & (0.00) & 1.00 & (0.00) \\
     & B & 0.50 & (0.50) & 0.50 & (0.50) & 1.00 & (0.00) & 1.00 & (0.00) & 1.00 & (0.00) & 1.00 & (0.00) \\
    QQ & H & 0.80 & (0.20) & 0.80 & (0.20) & 1.00 & (0.00) & 1.00 & (0.00) & 1.00 & (0.00) & 1.00 & (0.00) \\
     & B & 0.80 & (0.20) & 0.80 & (0.20) & 1.00 & (0.00) & 1.00 & (0.00) & 1.00 & (0.00) & 1.00 & (0.00) \\
    \bottomrule
\end{tabular}

    \label{tab:bfs}
\end{table}

Next, to show the benefit of the search heuristic, we compare against an ablation which ignores the heuristic and does breadth-first search.

At any point in the search process, there is a heap of search nodes, each of which has a lower and upper bound of $F_1$ scores reachable from it. Rather than using the lower bound from the abstract objective value, we instead evaluate the concrete program whose parameters are the midpoint of the hyper-rectangle of abstract parameters, to get a concrete objective value; this is a better lower-bound, and it is cheap to compute. The greatest of these lower bounds provides a lower bound on the optimal $F_1$ score, and the greatest of these upper bounds provides an upper bound on the optimal $F_1$ score. Note that if the lower and upper bounds are equal, they equal the true optimal $F_1$ score, and search terminates.

To make search tractable, on the CRIM13 benchmarks we consider only expressions with structural cost at most $4$. Note that this rules out ``$\operatorname{ite}$'' expressions, but performs well in practice. On the Quivr benchmarks, we bound the search space in the same way that their paper does, limiting to programs with at most 3 predicates, at most 2 of which have parameters.

We use 100 trajectories from each dataset. For CRIM13, these are randomly sampled from the training set. For Quivr, to ensure that we have some positive examples, because positives are very sparse on some tasks, we use 2 positive and 10 negative trajectories specially designated in the dataset, and the remaining 88 are sampled randomly from the training set.

Table~\ref{tab:bfs} shows, at different times during the search process, the best found $F_1$ score (the lower-bound of the interval), as well as the width of the interval of optimal $F_1$ scores. On most tasks, our approach (H) achieves higher $F_1$ scores more quickly than the ablation (B), as well as tighter intervals.

\section{Related Work}

\subsubsection*{Neurosymbolic Synthesis.}

There has been a great deal of recent interest in neurosymbolic synthesis~\cite{chaudhuri2021neurosymbolic}, including synthesis of functional programs~\cite{gaunt2016terpret,valkov2018houdini,shah2020learning}, reinforcement learning policies~\cite{anderson2020neurosymbolic,inala2020neurosymbolic},
programs for extracting data from unstructured text~\cite{chen2021web,ye2021optimal,chen2023data}, and programs that extract data from video trajectories~\cite{shah2020learning,bastani2021skyquery,mell2023synthesizing}. Some of these approaches have proposed pruning strategies based on monotonicity~\cite{chen2021web,mell2023synthesizing}, but for specific DSLs. NEAR is a general framework for neurosymbolic synthesis based on neural  heuristics~\cite{shah2020learning}; however, their approach is not guaranteed to synthesize optimal programs. To the best of our knowledge, our work proposes the first general framework for optimal synthesis of neurosymbolic programs.

\subsubsection*{Optimal Synthesis.}

More broadly, there has been recent interest in optimal synthesis~\cite{bornholt2016optimizing,smith2016mapreduce}, typically focusing on optimizing performance properties of the program such as running time rather than accuracy; superoptimization is a particularly well studied application~\cite{massalin1987superoptimizer,bansal2008binary,phothilimthana2016scaling,sasnauskas2017souper,mukherjee2020dataflow}. Our experiments demonstrate that our approach outperforms~\citet{bornholt2016optimizing}, a general framework for optimal synthesis based on SMT solvers. There has also been work on synthesizing a program that maximizes an objective (expressed as a neural network scoring function)~\cite{ye2021optimal},
but they do not consider real-valued constants, the quantitative objective is syntactic, not semantic, and abstract interpretation is only used for pruning according to the Boolean correctness specification, not the quantitative objective. Optimal synthesis has also been leveraged for synthesizing minimal guards for memory safety~\cite{dillig2014optimal},  chemical reaction networks~\cite{cardelli2017syntax}, and optimal layouts for quantum computing~\cite{tan2020optimal}.

\subsubsection*{Abstract Interpretation for Synthesis.}

There has been work on leveraging abstract interpretation for pruning portions of the search space in program synthesis~\cite{guria2023absynthe, so2017}, as well as using abstraction refinement~\cite{wang2017program}; however, these approaches target traditional synthesis. Rather than evaluating an abstract semantics on partial programs, \citet{wang2017program} constructs a data structure compactly representing concrete programs whose abstract semantics are compatible with the input-output examples. However, it is not obvious how their data structure (which targets Boolean specifications) can be adapted to our quantitative synthesis setting.

\subsubsection*{Abstract Interpretation for Planning.}

One line of work, initiated by the FF Planner~\cite{ffplanner}, uses abstract semantics to perform reachability analysis to prune invalid plans~\cite{gregory2012,zhi2022abstract,hoffmann2003metric}. However, other than pruning invalid plans, the reachability analysis is not used in the computation of the search heuristic. Instead, traditional heuristics such as $h_{\operatorname{max}}$ (which computes the shortest plan in a ``relaxed model'' that drops delete lists from the postconditions of abstract actions, and outputs the length of this plan) are used. In contrast, we use an abstract transformer for the objective function to provide a lower bound that is directly used as the search heuristic.

A second line of work~\cite{vegabrown2018,marthi2008angelic} considers computing optimal plans by underapproximating the cost function. They assume that the total cost of a plan equals the sum of the costs of the individual actions in that plan and then, given a lower bound on the cost of each action, simply sum these lower bounds to obtain a lower bound on the cost of the overall plan. This strategy makes strong assumptions about the structure of the overall cost function, whereas our abstract interpretation based approach requires no such assumptions.

Another key difference is that we are abstracting over real-valued parameters of partial programs, whereas the above approaches are abstracting over continuous states. Thus, our framework requires a way to iteratively refine the program space (specified by the “children” function), which is absent from their frameworks.

\section{Conclusion}

We have proposed a general framework for synthesizing programs with real-valued inputs and outputs, using $A^*$ search in conjunction with a search heuristic based on abstract interpretation. Our framework searches over a space of generalized partial programs, which represent sets of concrete programs, and uses the search heuristic to establish upper bounds on the objective value of a given generalized partial program. In addition, we propose a natural strategy for constructing abstract transformers for components with monotone semantics. If our algorithm returns a program, then this program is guaranteed to be optimal. Our experimental evaluation demonstrates that our approach is more scalable than existing optimal synthesis techniques.
Directions for future work include improving the scalability of our approach and applying it to additional synthesis tasks.

\section*{Acknowledgements}
We thank the anonymous reviewers for their helpful feedback. This work was supported in part by NSF Award CCF-1910769, NSF Award CCF-1917852, ARO Award W911NF-20-1-0080, and Amazon/ASSET Gift for Research in Trustworthy AI.

\section*{Data Availability Statement}
An artifact is available~\cite{artifact} with our implementation, which reproduces our experimental results (Figure~\ref{fig:exp:smt} and Table~\ref{tab:bfs}) and may be useful for performing synthesis on other trajectory datasets or implementing our algorithm for other DSLs.

\bibliographystyle{ACM-Reference-Format}
\bibliography{main}

\end{document}